\documentclass[prodmode,tissec]{acmsmall}

\usepackage{latexsym,amsmath,amssymb}
\usepackage{subfigure}
\usepackage{tabulary}
\usepackage{version}
\usepackage{graphicx}
\usepackage{paralist}
\usepackage{cite}
\usepackage{url}
\usepackage{color}
\usepackage{tikz}
\usetikzlibrary{positioning,arrows,automata,shapes,decorations.pathreplacing}

\newcommand{\set}[1]{\ensuremath{\left\{ #1 \right\}}}
\newcommand{\wsp}{{\sc WSP}}
\newcommand{\wspi}[1]{\textnormal{WSP}_{#1}}

\newcommand{\rrbac}{{\sf R$^2$BAC}}

\newcommand{\card}[1]{\left|#1\right|}
\newcommand{\children}[1]{\Delta(#1)}
\newcommand{\spbr}[1]{\sp{(#1)}}
\newcommand{\xspbr}[2]{#1\spbr{#2}}

\newcommand{\ceil}[1]{\lceil #1 \rceil}
\newcommand{\floor}[1]{\lfloor #1 \rfloor}

\newenvironment{newstuff}{\color{black}}{}
\newenvironment{oldstuff}{\color{red}}{}

\includeversion{newstuff}
\excludeversion{oldstuff}

\begin{document}

\markboth{J. Crampton et al.}{On the Parameterized Complexity and Kernelization of the Workflow Satisfiability Problem}

\title{On the Parameterized Complexity and Kernelization of the Workflow Satisfiability Problem}
\author{%
  JASON CRAMPTON\affil{Royal Holloway, University of London}
  GREGORY GUTIN\affil{Royal Holloway, University of London}
  ANDERS YEO\affil{University of Johannesburg}}

\begin{abstract}
A workflow specification defines a set of steps and the order in which those steps must be executed.
Security requirements may impose constraints on which groups of users are permitted to perform subsets of those steps.
A workflow specification is said to be satisfiable if there exists an assignment of users to workflow steps that satisfies all the constraints.
An algorithm for determining whether such an assignment exists is important, both as a static analysis tool for workflow specifications, and for the construction of run-time reference monitors for workflow management systems.
Finding such an assignment is a hard problem in general, but work by Wang and Li in 2010 using the theory of parameterized complexity suggests that efficient algorithms exist under reasonable assumptions about workflow specifications.
In this paper, we improve the complexity bounds for the workflow satisfiability problem.
We also generalize and extend the types of constraints that may be defined in a workflow specification and prove that the satisfiability problem remains fixed-parameter tractable for such constraints.
Finally, we consider preprocessing for the problem and prove that in an important special case, in polynomial time, we can reduce the given input into an equivalent one, where the number of users is at most the number of steps.
We also show that no such reduction exists for two natural extensions of this case, which bounds the number of users by a polynomial in the number of steps, provided a widely-accepted complexity-theoretical assumption holds.
\end{abstract}

\category{D4.6}{Operating Systems}{Security and Protection}[Access controls]
\category{F2.2}{Analysis of Algorithms and Problem Complexity}{Nonnumerical Algorithms and Problems}
\category{H2.0}{Database Management}{General}[Security, integrity and protection]

\terms{Algorithms, Security, Theory}


\acmformat{Crampton, J., Gutin, G., Yeo, A.  2013. On the Parameterized Complexity of the Workflow Satisfiability Problem.}

\begin{bottomstuff}
A preliminary version of this paper appeared in the Proceedings of CCS 2012.

Author's addresses:
J. Crampton, Department of Mathematics, Royal Holloway, University of London;
G. Gutin, Department of Computer Science, Royal Holloway, University of London;
A. Yeo, Department of Mathematics, University of Johannesburg.
\end{bottomstuff}

\maketitle

\section{Introduction}\label{sec:intro}

It is increasingly common for organizations to computerize their business and management processes.
The co-ordination of the tasks or steps that comprise a computerized business process is managed by a workflow management system (or business process management system).
Typically, the execution of these steps will be triggered by a human user, or a software agent acting under the control of a human user, and the execution of each step will be restricted to some set of authorized users.

A workflow typically specifies the steps that comprise a business process and the order in which those steps should be performed.
Moreover, it is often the case that some form of access control, often role-based, should be applied to limit the execution of steps to authorized users.
In addition, many workflows require controls on the users that perform groups of steps.
\begin{newstuff}The concept of a Chinese wall, for example, limits the set of steps that any one user can perform~\cite{BrNa89}, as does separation-of-duty, which is a central part of the role-based access control model~\cite{ansi-rbac04}.\end{newstuff}
Hence, it is important that workflow management systems implement security controls that enforce authorization rules and business rules, in order to comply with statutory requirements or best practice~\cite{BaBuKa10}.
It is these ``security-aware'' workflows that will be the focus of the remainder of this paper.

A simple, illustrative example for purchase order processing~\cite{cram:sacmat05} is shown in Figure~\ref{fig:example-workflow}.
In the first step of the workflow, the purchase order is created and approved (and then dispatched to the supplier).
The supplier will submit an invoice for the goods ordered, which is processed by the create payment step.
When the supplier delivers the goods, a goods received note (GRN) must be signed and countersigned.
Only then may the payment be approved and sent to the supplier.
Note that a workflow specification need not be linear: the processing of the GRN and of the invoice can occur in parallel, for example.

In addition to defining the order in which steps must be performed, the workflow specification includes rules to prevent fraudulent use of the purchase order processing system.
In our example, these rules take the form of constraints on users that can perform pairs of steps in the workflow: the same user may not sign and countersign the GRN, for example.
(We introduce more complex rules in Sections~\ref{sec:wsp} and~\ref{sec:fpt-organizational-constraints}.)

\begin{figure}[h]\centering
\hspace*{.075\textwidth}
\subfigure[Ordering on steps]{
\begin{tikzpicture}[->,.=stealth',node distance=8mm and 7mm,semithick,auto,state/.style={draw,circle,inner sep=3pt}]
  \node[state]  (t1)                      {$s_1$};
  \node[state]  (t2) [right=of t1]        {$s_2$};
  \node[state]  (t3) [above right=of t2]  {$s_3$};
  \node[state]  (t4) [below right=of t3]  {$s_4$};
  \node[state]  (t5) [above right=of t4]  {$s_5$};
  \node[state]  (t6) [below right=of t5]  {$s_6$};
  \path (t1) edge (t2)
   (t2) edge (t3)
   (t2) edge (t4)
   (t3) edge (t5)
   (t4) edge (t6)
   (t5) edge (t6);
\end{tikzpicture}}
\hfill
\subfigure[Constraints]{
\begin{tikzpicture}[-,node distance=8mm and 8mm,semithick,auto,state/.style={draw,circle,inner sep=3pt}]
  \node[state] (t1) {$s_1$};
  \node[state] (t2) [right=of t1] {$s_2$};
  \node[state] (t3) [above=of t2] {$s_3$};
  \node[state] (t4) [below=of t2] {$s_4$};
  \node[state] (t5) [right=of t3] {$s_5$};
  \node[state] (t6) [right=of t4] {$s_6$};
  \path (t1) edge [dotted] node {$=$} (t3)
        (t3) edge [dotted] node {$\ne$} (t5)
        (t1) edge [dotted] node[swap] {$\ne$} (t4)
        (t1) edge [dotted] node[swap] {$\ne$} (t2)
        (t4) edge [dotted] node {$\ne$} (t6);
\end{tikzpicture}}
\hspace*{.075\textwidth}

\subfigure[Legend]{\footnotesize\setlength{\extrarowheight}{2pt}
  \begin{tabular}{|ll|ll|ll|}
    \hline
    $s_1$ & create purchase order &     $s_4$ & create payment &     $\ne$ & different users must perform steps \\
    $s_2$ & approve purchase order &     $s_5$ & countersign GRN & $=$ & same user must perform steps \\
    $s_3$ & sign GRN &   $s_6$ & approve payment &  & \\
    \hline
    \multicolumn{6}{c}{} \\
  \end{tabular}}
\caption{A simple constrained workflow for purchase order processing}\label{fig:example-workflow}
\end{figure}
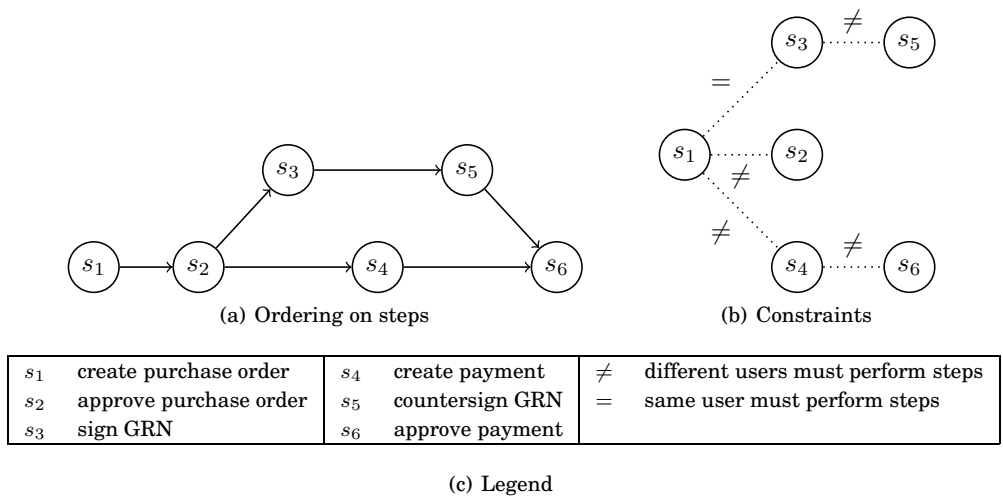

It is apparent that it may be impossible to find an assignment of authorized users to workflow steps such that all constraints are satisfied.
In this case, we say that the workflow specification is \emph{unsatisfiable}.
The {\sc Workflow Satisfiability Problem} (\wsp) is known to be NP-hard, even when the set of constraints only includes constraints that have a relatively simple structure (and that would arise regularly in practice).%
\footnote{In particular, the {\sc Graph $k$-Colorability} problem can be reduced to a special case of \wsp\ in which the workflow specification only includes separation-of-duty constraints~\cite{WangLi10}.}

It has been argued that it would be of practical value to be able to define constraints in terms of organizational structures,  rather than just the identity of particular users~\cite{WangLi10}.
One of the contributions of this paper is to introduce a model for hierarchical organizations based on the notion of equivalence classes and partition refinements.
We demonstrate how to construct an instance of our model from a management structure and illustrate why constraints defined over such models are of practical value.
\begin{newstuff}
The use of cardinality constraints in access control policies has also attracted considerable interest in the academic community~\cite{JoBeLaGh05,SaCoFeYo96,SiZu97}.
Cardinality constraints can encode a number of useful requirements that cannot be encoded using the constraints that have been used in prior work on \wsp.
A second contribution of this paper is to introduce \emph{counting} constraints for workflows---a natural extension of cardinality constraints---and to examine \wsp\ when such constraints form part of a workflow specification.
\end{newstuff}

\citeN{WangLi10} observed that the number of steps in a workflow is likely to be small relative to the size of the input to the workflow satisfiability problem.
This observation led them to study the problem using tools from parameterized complexity and to prove that the problem is fixed-parameter tractable for certain classes of constraints.
These results demonstrate that it is feasible to solve \wsp\ for many workflow specifications in practice.
However, Wang and Li also showed that for many types of constraints the problem is fixed-parameter intractable unless the parameterized complexity hypotheses $\text{FPT}\neq\text{W[1]}$ fails, which is highly unlikely.
(We provide a short introduction to parameterized complexity in Section~\ref{sec:fpt}.)
In this paper, we extend the results of Wang and Li in several different ways.
  \begin{enumerate}[1.]
    \item \begin{newstuff}First, we introduce the notion of counting constraints, a generalization of cardinality constraints, and extend the analysis of \wsp\ to include such constraints.\end{newstuff}
    \item Our second contribution is to introduce a new approach to \wsp, which makes use of a powerful, recent result in the area of exponential-time algorithms~\cite{BjHuKo09}.
        \begin{newstuff}
        We establish necessary and sufficient conditions on constraints that will admit the use of our approach.
        In particular, we show that counting constraints satisfy these conditions, as do the constraints considered by Wang and Li.
        \end{newstuff}
        This approach allows us to develop algorithms with a significantly better worst-case performance than those of Wang and Li.
        Moreover, we demonstrate that our result cannot be significantly improved, provided a well-known hypothesis about the complexity of solving 3-SAT holds.
    \item Our third extension to the work of Wang and Li is to define constraints in terms of hierarchical organizational structures and to prove, using our new technique, that \wsp\ remains fixed-parameter tractable in the presence of such hierarchical structures and hierarchy-related constraints.
    \item \begin{newstuff}Our fourth contribution is to instigate the systematic study of parameterized compression (also known as kernelization) of WSP instances.%
						  \footnote{Kernelization of WSP instances can be extremely useful in speeding up the solution of WSP: the compressed instance can be solved using any suitable algorithm (such as a SAT solver), not necessarily by an FPT algorithm.}
						  We show that a result of~\citeN[Theorem 3.3]{FeFrHeNaRo11} on a problem equivalent to a special case of WSP can be slightly extended and significantly improved using graph matchings.
						  We also prove that two natural further extensions of the result of Fellows et al. are impossible subject to a widely-accepted complexity-theoretical hypothesis.%
		  \end{newstuff}
\end{enumerate}

In the next section, we introduce the workflow satisfiability problem.
\begin{newstuff}
In Section~\ref{sec:wsp-fpt}, we provide a brief introduction to fixed-parameter tractability, prove a general result  characterizing the constraints for which WSP is fixed-parameter tractable, and apply this result to counting constraints.
In Section~\ref{sec:fpt-entailment-constraints} we extend the results of Wang and Li, by improving the complexity of the algorithms used to solve \wsp\ and by introducing constraints based on equivalence relations.
In Section~\ref{sec:fpt-organizational-constraints}, we introduce a model for an organizational hierarchy and a class of constraint relations defined in terms of such hierarchies.
We demonstrate that \wsp\ remains fixed-parameter tractable for workflow specifications that include constraints defined over an organizational hierarchy.
In Section~\ref{sec:kernelization}, we discuss kernelization of {\sc WSP} and prove that in an important special case, in polynomial time, we can transform the given input into an equivalent one, where the number of users is at most the number of steps.
We also show that no polynomial transformation exists for two natural extensions of this case, which bounds the number of users by a polynomial in the number of steps, unless a certain complexity-theoretical assumption fails.
The paper concludes with a summary of our contributions and discussions of related and future work.
\end{newstuff}

\section{The Workflow Satisfiability Problem}\label{sec:wsp}

In this section, we introduce our notation and definitions, derived from earlier work by~\citeN{cram:sacmat05} and~\citeN{WangLi10}, and then define the workflow satisfiability problem.

A partially ordered set (or poset) is a pair $(X,\leqslant)$, where $\leqslant$ is a reflexive, anti-symmetric and transitive binary relation defined over $X$.
If $(X,\leqslant)$ is a poset, then we write $x \parallel y$ if $x$ and $y$ are incomparable; that is, $x \not\leqslant y$ and $y \not\leqslant x$.
We may write $x \geqslant y$ whenever $y \leqslant x$.
We may also write $x < y$ whenever $x \leqslant y$ and $x \ne y$.
Finally, we will write $[n]$ to denote $\set{1,\dots,n}$.

\begin{definition}\label{def:workflow}
A \emph{workflow specification} is a partially ordered set of steps $(S,\leqslant)$. 
An \emph{authorization policy} for a workflow specification is a relation $A \subseteq S \times U$.
A \emph{workflow authorization schema} is a tuple $(S,U,\leqslant,A)$, where $(S,\leqslant)$ is a workflow specification and $A$ is an authorization policy. 
\end{definition}

If $s < s'$ then $s$ must be performed before $s'$ in any instance of the workflow; if $s \parallel s'$ then $s$ and $s'$ may be performed in either order.
\begin{newstuff}Our definition of workflow specification does not permit repetition of tasks (loops) or repetition of sub-workflows (cycles).\end{newstuff}
User $u$ is authorized to perform step $s$ only if $(s,u) \in A$.%
\footnote{In practice, the set of authorized step-user pairs, $A$, will not be defined explicitly.
          Instead, $A$ will be inferred from other access control data structures.
          In particular, \rrbac~--~the role-and-relation-based access control model of~\citeN{WangLi10}~--~introduces a set of roles $R$, a user-role relation ${\it UR} \subseteq U \times R$ and a role-step relation ${\it SA} \subseteq R \times S$ from which it is possible to derive the steps for which users are authorized.
          For all common access control policies (including \rrbac), it is straightforward to derive $A$.
          We prefer to use $A$ in order to simplify the exposition.}
We assume that for every step $s \in S$ there exists some user $u \in U$ such that $(s,u) \in A$.

\begin{definition}
Let $(S,U,\leqslant,A)$ be a workflow authorization schema.
A \emph{plan} is a function $\pi: S \rightarrow U$.
A plan $\pi$ is \emph{authorized} for $(S,U,\leqslant, A)$ if $(s,\pi(s)) \in A$ for all $s \in S$.
\end{definition}

\begin{newstuff}
The access control policy embodied in the authorization relation $A$ imposes restrictions on the users that can perform specific steps in the workflow.
A \emph{workflow authorization constraint} imposes restrictions on the execution of \emph{sets} of steps in a workflow.
A constraint is defined by some suitable syntax and its meaning is provided by the restrictions the constraint imposes on the users that execute the sets of steps defined in the constraint.
In other words, constraint satisfaction is defined with reference to a plan; a \emph{valid} plan is one that is authorized and allocates users in such a way that the constraint is satisfied.
A very simple example of a constraint is one requiring that steps $s$ and $s'$ are executed by different users.
Then a valid plan $\pi$ (with respect to this constraint) has the property that $\pi(s) \ne \pi(s')$.
A \emph{constrained workflow authorization schema} is a tuple $(S,U,\leqslant,A,C)$, where $C$ is a set of workflow constraints.\footnote{The set of constraints defines what has been called a \emph{history-dependent authorization policy}~\cite{BaBuKa12}; the relation $A$ defines a \emph{history-independent policy}.}
A plan is \emph{valid} for an authorization schema if it is authorized and satisfies all constraints in $C$.
We define particular types of constraints in Section~\ref{sec:constraint-types} and~\ref{sec:entailment-constraint-types}.
\end{newstuff}

We may now define the workflow satisfiability problem, as defined by~\citeN{WangLi10}.

\begin{center}
\fbox{%
      \begin{tabulary}{.95\textwidth}{@{}r<{~}@{}L@{}}
        \multicolumn{2}{@{}l}{\sc Workflow Satisfiability Problem (WSP)}\\
        \emph{Input:} & A constrained workflow authorization schema $(S,U,\leqslant,A,C)$\\
        \emph{Output:} & A valid plan $\pi : S \rightarrow U$ or an answer that there exists no valid plan
       \end{tabulary}%
      }
\end{center}
\begin{newstuff}
We will write $c$, $n$ and $k$ to denote the number of constraints, users and steps, respectively, in an instance of \wsp.
We will analyze the complexity of the workflow satisfiability problem in terms of these parameters.
\end{newstuff}

\begin{newstuff}
\subsection{Applications of \wsp}

An algorithm that solves \wsp\ can be used by a workflow management system in one of three ways, depending on how users are allocated to steps in an instance of the workflow.
Some systems allocate an authorized user to each step when a workflow instance is generated.
Other systems allocate users to only those steps that are ready to be performed in an instance of the workflow.
(A step is ready only if all its immediate predecessor steps have been completed.)
The third possibility is to allow users to select a step to execute from a pool of ready steps maintained by the workflow management system.

For the first type of system, it is important to know that a workflow is satisfiable and an algorithm that solves \wsp\ can simply be used as a static analysis tool.
The NP-hardness of the problem suggests that the worst-case run-time of such an algorithm will be exponential in the size of the input.
Hence, it is important to find an algorithm that is as efficient as possible.

For the second and third cases, the system must guarantee that the choice of user to execute a step (whether it is allocated by the system or selected by the user) does not prevent the workflow instance from completing.
This analysis needs to be performed each time a user is allocated to, or selects, a step in a workflow instance.
The question can be resolved by solving a new instance of \wsp, in which those steps to which users have been allocated are assumed to have a single authorized user (namely, the user that has been allocated to the task)~\cite[\S 3.2]{cram:sacmat05}.
Assuming that these checks should incur as little delay as possible, particularly in the case when users select steps in real time~\cite{KoSc08}, it becomes even more important to find an algorithm that can decide \wsp\ as efficiently as possible.

The definition of workflow satisfiability given above assumes that the set of users and the authorization relation are given.
This notion of satisfiability is appropriate when the workflow schema is designed ``in-house''.
A number of large information technology companies develop business process systems which are then configured by the end users of those systems.
Part of that configuration includes the assignment of users to steps in workflow schemas.
The developer of such a schema may wish to be assured that the schema is satisfiable for some set of users and some authorization relation, since the schema is of no practical use if no such user set and authorization relation exist.
The desired assurance can be provided by solving an instance of \wsp\ in which there are $k$ users, each of which is authorized for all steps.
The developer may also determine the minimum number of users required for a workflow schema to be satisfiable.
The minimum number must be between $1$ and $k$ and, using a binary search, can be determined by examining $\ceil{\log_2 k}$ instances of \wsp.
\end{newstuff}

\begin{newstuff}
\subsection{Constraint Types}\label{sec:constraint-types}

In this paper, we consider two forms of constraint: \emph{counting} constraints and \emph{entailment} constraints.
A counting constraint has the form $(t_\ell,t_r,S')$, where $1 \leqslant t_\ell \leqslant t_r \leqslant k$ and $S' \subseteq S$.
A counting constraint is a generalization of the cardinality constraints introduced in the RBAC96 model~\cite{SaCoFeYo96} and widely adopted by subsequent access control models~\cite{ansi-rbac04,BeboFe01,JoBeLaGh05}.

A plan $\pi : S \rightarrow L$ satisfies counting constraint $(t_\ell,t_r,S')$ if a user performs either no steps in $S'$ or between $t_\ell$ and $t_r$ steps.
In other words, no user is assigned to more than $t_r$ steps in $S'$ and each user (if involved in the execution of steps in $S'$) must perform at least $t_\ell$ steps.
Many requirements give rise to counting constraints of the form $(t,t,S')$, which we will abbreviate to $(t,S')$.
A number of requirements that arise in the literature and in practice can be represented by counting constraints.
  \begin{description}
    \item[Separation of duty] The constraint $(1,\set{s',s''})$ requires that no user executes both $s'$ and $s''$.  More generally, the constraint $(1,\card{S'}-1,S')$ requires that no user executes all the steps in $S'$.
    \item[Binding of duty] The constraint $(2,\set{s',s''})$ requires that the same user executes both $s'$ and $s''$.  More generally, the constraint $(\card{S'},S')$ requires that all steps in $S'$ are executed by the same user.
    \item[Division of duty] The constraint $(\floor{\card{S'}/v},\ceil{\card{S'}/v},S')$ requires that the steps in $S'$ are split as equally as possible between $v$ different users. The special case $(1,S')$ requires that a different user performs each step in $S'$.
    \item[Threshold constraints] The constraint $(1,t,S')$ requires that no user executes more than $t$ steps in $S'$.%
          \footnote{These constraints are similar in structure and analogous in meaning to SMER (statically, mutually-exclusive, role) constraints~\cite{LiTrBi07}; the SMER constraint $(t,S')$ requires that no user is authorized for $t$ or more roles in the set of roles $S'$.  These constraints are also similar to the cardinality constraints defined in RBAC96~\cite{SaCoFeYo96}.}
    \item[Generalized threshold constraints] The constraint $(t_\ell,t_r,S')$ requires that each user (involved in the execution of steps in $S'$) performs between $t_\ell$ and $t_r$ of those steps.
  \end{description}

Counting constraints are not able to encode certain types of requirements.
For this reason, we also consider entailment constraints, which have the form $(\rho,S',S'')$, where $\rho \subseteq U \times U$ and $S',S'' \subseteq S$.
A plan $\pi$ satisfies entailment constraint $(\rho,S',S'')$ if and only if there exists $s' \in S'$ and $s'' \in S''$ such that
$(\pi(s'),\pi(s'')) \in \rho$.
A plan $\pi$ satisfies a set of constraints $C$ (which may be a mixture of counting and entailment constraints) if $\pi$ satisfies each constraint in $C$.

Counting constraints represent ``universal'' restrictions on the execution of steps (in the sense that every user in a plan must satisfy the requirement stipulated).
In contrast, entailment constraints are ``existential'' in nature: they require the existence of a pair of steps for which a condition on the two users who execute those steps (defined by the binary relation $\rho$) is satisfied.

We could write $\delta$ to the denote the diagonal relation $\set{(u,u) : u \in U}$ and $\overline{\delta}$ to denote $(U \times U) \setminus \delta$.
However, we will prefer to use the less formal, but more intuitive, notation $(\ne,S',S'')$ and $(=,S',S'')$ to denote the constraints $(\overline{\delta},S',S'')$ and $(\delta,S',S'')$, respectively.

There are some requirements that can be represented by a counting constraint or an entailment constraint.
The counting constraint $(1,\set{s_1,s_2})$, for example, is satisfied by plan $\pi$ if and only if the entailment constraint $(\ne,\set{s_1},\set{s_2})$ is satisfied.
We say that two constraints $\gamma$ and $\gamma'$ are \emph{equivalent} if a plan $\pi$ satisfies $\gamma$ if and only if it satisfies $\gamma'$.
Thus $(1,\set{s_1,s_2})$ is equivalent to $(\ne,\set{s_1},\set{s_2})$.
Similarly, $(2,\set{s_1,s_2})$ is equivalent to $(=,\set{s_1},\set{s_2})$.
Nevertheless, there is no counting constraint (or set of such constraints) that is equivalent to $(=,S_1,S_2)$.
Equally, there is no entailment constraint (or set of such constraints) that is equivalent to $(t,S')$.
\end{newstuff}

\begin{newstuff}
\subsection{Entailment Constraint Subtypes}\label{sec:entailment-constraint-types}

Previous work on workflow satisfiability has not considered counting constraints.
Moreover, our definition of entailment constraint is more general than prior definitions.
Thus, we study more general constraints for \wsp\ than have been investigated before.

\citeN{cram:sacmat05} defined entailment constraints in which $S_1$ and $S_2$ are singleton sets: we will refer to constraints of this form as \emph{Type $1$} constraints; for brevity we will write $(\rho,s_1,s_2)$ for the Type 1 constraint $(\rho,\set{s_1},\set{s_2})$.
\citeN{WangLi10} defined constraints in which at least one of $S_1$ and $S_2$ is a singleton set: we will refer to constraints of this form as \emph{Type $2$} constraints and we will write $(\rho,s_1,S_2)$ in preference to $(\rho,\set{s_1},S_2)$.
The Type 2 constraint $(\rho,s_1,S_2)$ is equivalent to $(\rho,S_2,s_1)$ if $\rho$ is symmetric, in which case we will write $(\rho,s_1,S_2)$ in preference to $(\rho,S_2,s_1)$.
Note that both $\delta$ and $\overline{\delta}$ are symmetric binary relations.
Constraints in which $S_1$ and $S_2$ are arbitrary sets will be called \emph{Type $3$} constraints.

We note that Type 1 constraints can express requirements of the form described in Section~\ref{sec:intro}, where we wish to restrict the combinations of users that perform pairs of steps.
The plan $\pi$ satisfies constraint $(=,s,s')$, for example, if the same user is assigned to both steps by $\pi$, and satisfies constraint $(\ne,s,s')$ if different users are assigned to $s$ and $s'$.

Type 2 constraints provide greater flexibility, although Wang and Li, who introduced these constraints, do not provide a use case for which such a constraint would be needed.
However, there are forms of separation-of-duty requirements that are most naturally encoded using Type 3 constraints.
Consider, for example, the requirement that a set of steps $S' \subseteq S$ must not all be performed by the same user~\cite{ArGiPo09}.
We may encode this as the constraint $(\ne,S',S')$, which is satisfied by a plan $\pi$ only if there exists two steps in $S'$ that are allocated to different users by $\pi$.%
\footnote{%
It is interesting to note that a Type 3 constraint $(\ne,S',S'')$ can be encoded as a Type 2 constraint, thereby providing retrospective motivation for the introduction of Type 2 constraints by Wang and Li.
In particular, we may encode $(\ne,S',S'')$ as $(\ne,s,S' \cup S'' \setminus \set{s})$ for some $s \in S' \cup S''$.
The equivalence of these two constraints is left as an exercise for the interested reader.
(Note that we may also encode this requirement as the counting constraint $(1,\card{S'}-1,S')$.)}
The binding-of-duty constraint $(=,S',S'')$ cannot be directly encoded using Type 2 constraints or counting constraints.

Now consider a business rule of the form ``two steps must be performed by members of the same organizational unit''.
The constraint relations $=$ and $\ne$ do not allow us to define such constraints.
In Section~\ref{sec:fpt-entailment-constraints}, we model constraints of this form using equivalence relations defined on the set of users.
In Section~\ref{sec:fpt-organizational-constraints}, we introduce a model for hierarchical organizational structures, represented in terms of multiple, related equivalence relations defined on the set of users.
We then consider constraints derived from such equivalence relations and the complexity of \wsp\ in the presence of such constraints.

Henceforth, we will write \wsp$(\rho_1,\dots,\rho_t)$ to denote a special case of \wsp\ in which all constraints have the form $(\rho_i,S',S'')$ for some $\rho_i \in \set{\rho_1,\dots,\rho_t}$ and for some $S',S'' \subseteq S$.
We will write $\wspi{i}(\rho_1,\dots,\rho_t)$ to denote a special case of \wsp$(\rho_1,\dots,\rho_t)$, in which there are no constraints of Type $j$ for $j > i$.
So $\wspi{1}(=,\ne)$, for example, indicates an instance of \wsp\ in which all constraints are of Type 1 and only includes constraints of the form $(=,s_1,s_2)$ or $(\ne,s_1,s_2)$ for some $s_1,s_2 \in S$.
For ease of exposition, we will consider counting constraints and entailment constraints separately.
Our results, however, hold when a workflow specification includes both types of constraints.
\end{newstuff}

\begin{newstuff}
\section{$\wspi{}$ and Fixed-Parameter Tractability}\label{sec:wsp-fpt}

In order to make the paper self-contained, we first provide a short overview of parameterized complexity, what it means for a problem to be fixed-parameter tractable, and summarize the results obtained by Wang and Li for $\wspi{}$.
We then introduce the notion of an eligible set of steps.
The identification of eligible sets is central to our method for solving WSP.
In the final part of this section, we state and prove a ``master'' theorem from which a number of useful results follow as corollaries.
The master theorem also provides useful insights into the structure of constraints that will result in instances of WSP that are fixed-parameter tractable.
\end{newstuff}

\subsection{Parameterized Complexity}\label{sec:fpt}

A na\"ive approach to solving \wsp\ would consider every possible assignment of users to steps in the workflow.
There are $n^k$ such assignments if there are $n$ users and $k$ steps, so an algorithm of this form would have (worst-case) complexity $O(cn^k)$, where $c$ is the number of constraints.
Moreover, Wang and Li showed that \wsp\ is NP-hard, by reducing {\sc Graph $k$-Colorability} to \wsp$(\ne)$~\cite[Lemma 3]{WangLi10}.
In short, \wsp\ is hard to solve in general.
The importance of finding an efficient algorithm for solving \wsp\ led Wang and Li to look at the problem from the perspective of parameterized complexity~\cite[\S 4]{WangLi10}.

Suppose we have an algorithm that solves an NP-hard problem in time $O(f(k)n^d)$, where $n$ denotes the size of the input to the problem, $k$ is some (small) parameter of the problem, $f$ is some function in $k$ only, and $d$ is some constant (independent of $k$ and $n$).
Then we say the algorithm is a \emph{fixed-parameter tractable} (FPT) algorithm.
If a problem can be solved using an FPT algorithm then we say that it is an \emph{FPT problem} and that it belongs to the class FPT.

Wang and Li showed, using an elementary argument, that $\wspi{2}(\ne)$ is FPT and can be solved in time $O(k^{k+1}N)$, where $N$ is the size of the entire input to the problem~\cite[Lemma 8]{WangLi10}.
They also showed that \mbox{$\wspi{2}(\ne,=)$} is FPT~\cite[Theorem 9]{WangLi10}, using a rather more complex approach: specifically, they constructed an algorithm that runs in time \mbox{$O(k^{k+1}(k-1)^{k2^{k-1}}N)$}; it follows that $\wspi{2}(=,\ne)$ is FPT.

When the runtime $O(f(k)n^d)$ is replaced by the much more powerful $O(n^{f(k)})$, we obtain the class XP, where each problem is polynomial-time solvable for any fixed value of $k$.
There is an infinite collection of parameterized complexity classes, $\text{W}[1], \text{W}[2],\dots$, with $\text{FPT} \subseteq \text{W[1]} \subseteq \text{W[2]} \subseteq \dots \subseteq \text{XP}$.
\begin{newstuff}
Informally, a parameterized problem belongs to the complexity class W[$i$] if there exists an FPT algorithm that transforms every instance of the problem into an instance of {\sc Weighted Circuit Satisfiability} for a circuit of weft $i$.
It can be shown that FPT is the class W[$0$].
The problems {\sc Independent Set} and {\sc Dominating Set} are in W[1] and W[2], respectively.
It is widely-believed and often assumed that $\text{FPT} \neq \text{W}[1]$.
For a more formal introduction to the W family of complexity classes, see~\citeN{FlumGrohe06}.
\end{newstuff}

\citeN[Theorem 10]{WangLi10} proved that WSP (for arbitrary relations defined on the user set) is W[1]-hard in general, using a reduction from {\sc Independent Set}.
By definition, FPT is a subset of W[1] and a parameterized analog of Cook's Theorem~\cite{DowneyFellows99} as well as the Exponential Time Hypothesis~\cite{FlumGrohe06,ImPaZa01} strongly support the widely held view that FPT is not equal to W[1].
One of the main contributions of this paper is to extend the set of special cases of WSP that are known to be FPT.

Henceforth, we often write $\widetilde{O}(T)$ instead of $O(T \log^d T)$ for any constant $d$.
That is, we use the notation $\widetilde{O}$ to suppress polylogarithmic factors.
This notation is often used in the literature on algorithms---see, for example, \citeN{BjHuKo09} and \citeN{KauKriRon03}---to avoid cumbersome runtime bounds.

\subsection{Eligible Sets}

The basic idea behind our results is to construct a valid plan by partitioning the set of steps $S$ into blocks of steps, each of which is allocated to a single (authorized) user.
More formally, let $\pi$ be a valid plan for a workflow \mbox{$(S,U,\leqslant,A,C)$} and define an equivalence relation $\sim_\pi$ on $S$, where $s \sim_\pi s'$ if and only if $\pi(s) = \pi(s')$.
We denote the set of equivalence classes of $\sim_\pi$ by $S/\pi$ and write $[s]_{\pi}$ to denote the equivalence class containing $s$.
\begin{newstuff}An equivalence class in $S/\pi$ comprises the set of steps that are assigned to a single user by plan $\pi$.\end{newstuff}
It is easy to see that there are certain ``forbidden'' subsets $S'$ of $S$ for which there cannot exist a valid plan $\pi$ such that $S' \in S/\pi$.
Consider, for example, the constraint $(\ne,s,s')$: then, for any valid plan $\pi$, it must be the case that $[s]_\pi \ne [s']_\pi$; in other words, there does not exist a valid plan $\pi$ such that $\set{s,s'} \in S/\pi$.
This motivates the following definition.

\begin{newstuff}
\begin{definition}
Given a workflow $(S,U,\leqslant,A,C)$ and a constraint $\gamma \in C$, a set $F \subseteq S$ is \emph{$\gamma$-ineligible} if any plan $\pi : S \rightarrow U$ such that $F \in S/\pi$ violates $\gamma$.
We say $F$ is \emph{eligible} if and only if it is not ineligible.
We say $F \subseteq S$ is \emph{$C$-ineligible} or simply \emph{ineligible} if $F$ is $\gamma$-ineligible for some $\gamma \in C$.
\end{definition}

A necessary condition for a valid plan is that no equivalence class is an ineligible set; equivalently, every equivalence class in a plan must be an eligible set.
For many constraints $\gamma$, we can determine whether $F \subseteq S$ is $\gamma$-ineligible or not in time polynomial in the number of steps.
Consider, for example, the requirement that no user executes more than $t$ steps: then $F \subseteq S$ is eligible if and only if $\card{F} \leqslant t$.
Similarly, we can test for the ineligibility of $F$ with respect to $(\ne,\set{s_1,s_2})$ by determining whether $F \supseteq \set{s_1,s_2}$.

\begin{definition}
We say a constraint $\gamma$ is \emph{regular} if any plan $\pi$ in which each equivalence class $[s]_\pi$ is an eligible set satisfies $\gamma$.
\end{definition}

The regularity of a constraint is a sufficient condition to guarantee that we can construct a valid plan using eligible sets.
With one exception, all constraints we consider are regular.

\begin{proposition}\label{pro:regular-constraints}
All counting constraints are regular and all entailment constraints of the form $(\ne,S_1,S_2)$ are regular.
Entailment constraints of the form $(=,S_1,S_2)$ are regular if at least one of $S_1$ and $S_2$ is a singleton set.
\end{proposition}

\begin{proof}
The result is trivial for counting constraints.

Given an entailment constraint $(\ne,S_1,S_2)$, a plan $\pi$ in which all equivalence classes are eligible, and $[s]_\pi$ for some $s \in S_1 \cup S_2$, we have that $[s]_\pi \not\supseteq S_1 \cup S_2$ (since, by assumption, $[s]_\pi$ is eligible).
Hence, there exists an element $s' \in S_1 \cup S_2$ with $s' \not\in [s]_\pi$.
Since the equivalence classes in $S/\pi$ form a partition of $S$, there exists an equivalence class $[s']_\pi \ne [s]_\pi$.
Hence, the constraint is satisfied (since each equivalence class is assigned to a different user).
Thus the constraint is regular.

We demonstrate, by exhibiting a counterexample, that a partition of $S$ into eligible sets does not guarantee the satisfaction of a Type 3 constraint of the form $(=,S_1,S_2)$.
Consider, for example, $S = \set{s_1,s_2,s_3,s_4}$ and the constraint $(=,\set{s_1,s_2},\set{s_3,s_4})$.
Then $\set{s_1},\dots,\set{s_4}$ are eligible sets, but a plan in which $u_i$ is assigned to $s_i$ is not valid.

Finally, consider the Type 2 constraint $(=,s_1,S_2)$.
Any eligible set for this constraint that contains $s_1$ must contain an element of $S_2$.
Hence a partition of $S$ into eligible sets ensures that the constraint will be satisfied (and hence that the constraint is regular).
\end{proof}

\end{newstuff}

\subsection{Reducing \wsp\ to {\sc Max Weighted Partition}}

\begin{newstuff}
We now state and prove our main result.
We believe this result subsumes existing results in the literature on the complexity of WSP.
Moreover, the result considerably enhances our understanding of the types of constraints that can be used in a workflow specification if we wish to preserve fixed-parameter tractability of WSP.
We explore the consequences and applications of our result in Sections~\ref{sec:fpt-entailment-constraints} and~\ref{sec:fpt-organizational-constraints}.

\begin{theorem}\label{thm:master}
Let $W = (S,U,\leqslant,A,C)$ be a workflow specification such that
  \begin{inparaenum}[(i)]
    \item each constraint $\gamma$ is regular and
    \item there exists an algorithm that can determine whether $F \subseteq S$ is $\gamma$-eligible in time polynomial in $k$.
  \end{inparaenum}
Then the workflow satisfiability problem for $W$ can be solved in time $\widetilde{O}(2^k(c + n^2))$.
\end{theorem}
\end{newstuff}

The proof of this result reduces an instance of WSP to an instance of the {\sc Max Weighted Partition} problem, which, by a result of \citeN{BjHuKo09}, is FPT.
We state the problem and the relevant result, before proving Theorem~\ref{thm:master}.

\begin{center}
\fbox{%
        \begin{tabulary}{.95\textwidth}{@{}r<{~}@{}L@{}}
          \multicolumn{2}{@{}l}{\sc Max Weighted Partition}\\
          \emph{Input:} & A set $S$ of $k$ elements and $n$ functions $\phi_i$, $i\in [n]$, from $2^S$ to integers from the range $[-M,M]$ $(M\ge 1)$.\\
          \emph{Output:} & An $n$-partition $(F_1,\ldots ,F_n)$ of $S$ that maximizes $\sum_{i=1}^n\phi_i(F_i).$
        \end{tabulary}%
      }
\end{center}

\begin{theorem}[\citeN{BjHuKo09}]\label{thm:MWP}
{\sc Max Weighted Partition} can be solved in time $\widetilde{O}(2^{k}n^2M)$.
\end{theorem}

\begin{newstuff}
\begin{proof}[of Theorem~\ref{thm:master}]
We construct a binary matrix with $n$ rows (indexed by elements of $U$) and $2^k$ columns (indexed by elements of $2^S$): every entry in the column labeled by the empty set is defined to be $1$; the entry indexed by $u \in U$ and $F \subseteq S$ is defined to be $0$ if and only if $F \ne \emptyset$ is $C$-ineligible or there exists $s \in F$ such that $(s,u) \not\in A$.
In other words, the non-zero matrix entry indexed by $u$ and $F$ defines a $C$-eligible set and $u$ is authorized for all steps in $F$, and thus represents a set of steps that could be assigned to a single user in a valid plan.

The matrix defined above encodes a family of functions $\set{\phi_u}_{u \in U}$, $\phi_u : 2^S \rightarrow \set{0,1}$.
We now solve {\sc Max Weighted Partition} on input $S$ and $\set{\phi_u}_{u \in U}$.
Given that $\phi_u(F) \leqslant 1$, $\sum_{u \in U} \phi_u(F_u) \leqslant n$, with equality if and only if we can partition $S$ into different $C$-eligible blocks and assigned them to different users.
Since each $\gamma$ is regular, $W$ is satisfiable if and only if MWP returns a partition having weight $n$.

We now consider the complexity of the above algorithm.
By assumption, we can identify the ineligible sets in  $O(c \cdot k^d \cdot 2^k) = \widetilde{O}(c 2^k)$ time for some integer $d$ independent of $k$ and $c$.
And we can check whether a user is authorized for all steps in $F \subseteq S$ in $O(k)$ time.
Thus we can construct the matrix in $O(2^k \cdot n \cdot k) = \widetilde{O}(2^k n)$ time.
Finally, we can solve {\sc Max Weighted Partition} in $\widetilde{O}(2^k n^2)$ time.
Thus, the total time required to solve WSP for $W$ is $\widetilde{O}(2^k (c + n + n^2)) = \widetilde{O}(2^k(c + n^2))$.
\end{proof}
\end{newstuff}

\begin{newstuff}
\begin{theorem}\label{thm:counting-constraints-fpt}
WSP is FPT for any workflow specification in which all the constraints are counting constraints.
\end{theorem}

\begin{proof}
A plan $\pi : S \rightarrow L$ satisfies counting constraint $\gamma = (t_\ell,t_r,S')$ if a user performs either no steps in $S'$ or between $t_\ell$ and $t_r$ steps.
Hence, $F \subseteq S$ is eligible if and only if $t_\ell \leqslant \card{F} \leqslant t_r$, a test that can clearly be evaluated in $O(k)$ time.
The result now follows by Proposition~\ref{pro:regular-constraints} and Theorem~\ref{thm:master}.
\end{proof}

While the above result appears easy to state and prove, nothing was known about the complexity of incorporating such constraints into workflow specifications.
Moreover, counting constraints can be used to encode (Type 1) entailment constraints of the form $(\ne,s_1,s_2)$ and $\wspi{1}(\ne)$ is known to be NP-complete~\cite[Lemma 3]{WangLi10}.
Finally, counting constraints can encode requirements that cannot be expressed using entailment constraints.
Hence, WSP in the presence of counting constraints is at least as hard as $\wspi{1}(\ne)$.
Therefore, there is no immediate reason to suppose that WSP for counting constraints would be FPT.
In short, Theorem~\ref{thm:counting-constraints-fpt} is non-trivial, thus demonstrating the power of Theorem~\ref{thm:master}.

At first glance, it is perhaps surprising to discover that counting constraints have no effect on the fixed-parameter tractability of WSP.
However, on further reflection, the structure of the proof of Theorem~\ref{thm:master} suggests that \emph{any} constraint whose satisfaction is phrased in terms of the steps that a single user performs can be incorporated into a workflow specification without comprising fixed-parameter tractability.

It also becomes apparent that there are certain constraints whose inclusion may cause problems.
Any constraint whose satisfaction is defined in terms of the set of users that perform a set of steps may be problematic.
The requirement that a workflow be performed by at least three users, for example, cannot be encoded using the counting or entailment constraints we have defined in this paper.
Moreover, it is difficult to envisage an eligibility test for such a constraint and, if such a test exists, whether it can be evaluated in time polynomial in $k$.
However, we can express a constraint of this form as a counting constraint such that the original constraint is satisfied if the  counting constraint is satisfied.
Specifically, the requirement that a set of $S'$ steps be performed by at least $t$ users can be enforced by ensuring that each user performs no more than $(\card{S'}-1)/(t-1)$ steps.\footnote{Of course, this means that certain plans that do not violate the original requirement are invalid.  That is, the counting constraint ``over-enforces'' the original requirement.  See the work of \citeN{LiTrBi07} for further details on constraint rewriting of this nature.}
\end{newstuff}

\begin{newstuff}
\section{Entailment Constraints}\label{sec:fpt-entailment-constraints}

In this section we focus on workflow specifications that include only entailment constraints.
In doing so, we demonstrate further the power of Theorem~\ref{thm:master}.
We also show that the time complexity obtained in Theorem~\ref{thm:master} cannot be significantly improved even for a very special case of WSP.
We conclude with a discussion of and comparison with related work.

\subsection{$\wspi{}(\ne)$}

By Proposition~\ref{pro:regular-constraints}, any constraint $\gamma$ of the form $(\ne,S_1,S_2)$ is regular.
Moreover, there exists an easy test to determine whether $F \subseteq S$ is $\gamma$-ineligible.
Specifically, $F$ is $\gamma$-ineligible if and only if $F \supseteq S_1 \cup S_2$, since any plan that allocated a single user to the steps in $F$ would be invalid.
Hence, we can determine in time polynomial in the sizes of $F$, $S_1$ and $S_2$ (that is, in $k$) the eligibility of $\gamma$.

\begin{theorem}\label{thm:1level}
$\wspi{}(\ne)$ can be solved in time $\widetilde{O}(2^k(c + n^2))$.
\end{theorem}

\begin{proof}
The result follows from Theorem~\ref{thm:master} and the fact that every constraint is regular and the eligibility of any constraint can be determined in time polynomial in $k$.
\end{proof}

Our next result asserts that it is impossible, assuming the well-known \emph{Exponential Time Hypothesis}~\cite{ImPaZa01}, to improve this result to any significant degree.

\begin{center}
\fbox{%
        \begin{tabulary}{.95\textwidth}{@{}L@{}}
          \sc Exponential Time Hypothesis \\
          There exists a real number $\epsilon > 0$ such that {\sc 3-SAT} cannot be solved in time $O(2^{\epsilon n})$, where $n$ is the number of variables.
        \end{tabulary}%
      }
\end{center}

\begin{theorem}\label{thm:lb}
Even if there are just two users, $\wspi{2}(\ne)$ cannot be solved in time $\widetilde{O}(2^{\epsilon k})$ for some positive real $\epsilon$, where $k$ is the number of steps, unless the Exponential Time Hypothesis fails.
\end{theorem}
The proof of this result can be found in the appendix.

\subsection{$\wspi{}(=)$}

Given a constraint $\gamma$ of the form $(=,S_1,S_2)$, any set $F$ that contains $S_1$ but no element of $S_2$ is ineligible; equally, any set $F$ that contains $S_2$ but no element of $S_1$ is ineligible.
Hence, we can determine $\gamma$-ineligibility in time polynomial in $k$ (as we only require subset inclusion and intersection operations on sets whose cardinalities are no greater than $k$).
However, a constraint $\gamma$ of the form $(=,S_1,S_2)$ is not necessarily regular (Proposition~\ref{pro:regular-constraints}).
Nevertheless, we have the following result.

\begin{theorem}\label{thm:1level-type3}
$\wspi{2}(=)$ can be solved in time $\widetilde{O}(2^k(c+n^2))$, where $k$ is the number of steps, $c$ is the number of constraints and $n$ is the number of users.
$\wspi{}(=)$ can be solved in time $\widetilde{O}(2^{k+c}(c+n^2))$.
\end{theorem}

\begin{proof}
The first result follows immediately from Theorem~\ref{thm:master} and Proposition~\ref{pro:regular-constraints}, since the latter result asserts that constraints of the form $(=,s_1,S_2)$ are regular.

To obtain the second result, we rewrite a Type 3 constraint $(=,S_1,S_2)$ as two Type $2$ constraints, at the cost of introducing additional workflow steps.
Specifically, we replace a Type $3$ constraint $(=,S_1,S_2)$ with the constraints $(=,S_1,s_{\rm new})$ and $(=,s_{\rm new},S_2)$, where $s_{\rm new}$ is a ``dummy'' step.
Every user is authorized for $s_{\rm new}$.
Observe that if we have a plan that satisfies $(=,S_1,S_2)$ then there exists a user $u$ and steps $s_1 \in S_1$ and $s_2 \in S_2$ such that $\pi(s_1) = \pi(s_2)$.
Hence we can find a plan that satisfies $(=,S_1,s_{\rm new})$ and $(=,s_{\rm new},S_2)$: specifically, we extend $\pi$ by defining $\pi(s_{\rm new}) = u$.
Similarly, if we have a plan that satisfies $(=,S_1,s_{\rm new})$ and $(=,s_{\rm new},S_2)$ then there exists a user $u$ and steps $s_1$ and $s_2$ such that $u = \pi(s_{\rm new}) = \pi(s_1) = \pi(s_2)$ and we may construct a valid plan for $(=,S_1,S_2)$.

The rewriting of a (Type 3) constraint $(=,S_1,S_2)$ requires the replacement of one Type 3 constraint with two Type 2 constraints and the creation of one new step.
In other words, we can derive an equivalent instance of $\wspi{2}(=)$ having no more than $c$ additional constraints and no more than $c$ additional steps.
Since Type 2 constraints are regular, the result now follows by Theorem~\ref{thm:1level}.
\end{proof}

\begin{corollary}\label{cor:type3-equality-constraints-fpt}
$\wspi{}(=)$ is FPT.
\end{corollary}

\begin{proof}
We may assume without loss of generality that $S_1 \cap S_2 = \emptyset$: the constraint is trivially satisfied if there exists $s \in S_1 \cap S_2$, since we assume there exists at least one authorized user for every step.
Hence, the number of constraints having this form is no greater than $\sum_{j=1}^k \binom{k}{j}2^{k-j} = 3^k$.
Hence, $\wspi{}(=)$ is FPT, since we can replace $2^{k+c}$ in the run-time by $2^{k+3^k}$, as required.
\end{proof}

\subsection{$\wspi{}(=,\ne)$ and Related Work}

We can combine the results of the previous sections in a single theorem.
Clearly, we could also incorporate counting constraints into this result.

\begin{theorem}
$\wspi{}(=,\ne)$ can be solved in time $\widetilde{O}(2^{k+c}(c+n^2))$.
\end{theorem}

The special case of the workflow satisfiability problem $\wspi{2}(\ne)$ was studied by Wang and Li from the perspective of fixed-parameter tractability; the complexity of their algorithm is $O(k^{k+1} N) = 2^{O(k \log k)} N$, where $N$ is the size of the input~\cite[Lemma 8]{WangLi10}.
\citeN{FeFrHeNaRo11} considered the fixed-parameter tractability of a special case of the \emph{constraint satisfaction problem}~\cite{Tsang93} in which all constraints have the same form; with these restrictions, the constraint satisfaction problem is identical to $\wspi{1}(\ne)$.
The algorithm of Fellows {et al.} has complexity $O(k! k n) = 2^{O(k \log k)} n$, where $n$ is the number of users~\cite[Theorem 3.1]{FeFrHeNaRo11}.
Our algorithm has complexity $\widetilde{O}(2^k(c+n^2)) = O(2^{k+d \log k}(c + n^2))$, where $d = O(1)$, which represents a considerable improvement in the term in $k$.

More significantly, \citeN[Theorem 9]{WangLi10} showed that $\wspi{2}(\ne,=)$ is FPT; the complexity of their algorithm is \mbox{$O(k^{k+1}(k-1)^{k2^{k-1}}n)$}.
Our algorithm to solve $\wspi{2}(=,\ne)$ retains the complexity $\widetilde{O}(2^k(c + n^2))$, which is clearly a substantial improvement on the result of Wang and Li.
Finally, we note that our results are the first to consider Type 3 constraints.

\subsection{Constraints Based on Equivalence Relations}

The work of~\citeN[\S 2]{cram:sacmat05} and of~\citeN[Examples 1, 2]{WangLi10} has noted that a constraint of practical interest is that users performing two steps must be from the same department.\footnote{However, little is known about the complexity of the \wsp\ when such constraints are used, a deficiency we address in the next section.}
In the workflow illustrated in Figure~\ref{fig:example-workflow} one might require, for example, that the two users who perform steps $s_3$ and $s_5$ belong to the same department.
Note, however, that we will still require that these two users be different.
More generally, we might wish to insist that the user who approves the purchase order (step $s_2$) belongs to the same department as the user who creates the order (step $s_1$).

In short, there are many practical situations in which some auxiliary information defines an equivalence relation on the set of users (membership of department, for example) where we may wish to require that two steps are performed by users belonging to either the same equivalence class or to different equivalence classes.
In this section, we introduce two relations that allow us to model organizational structures, in which users are partitioned (possibly at several levels) into different organizational units, such as departments.

\begin{newstuff}
Given an equivalence relation $\sim$ on $U$, a plan $\pi$ satisfies the constraint $(\sim,S_1,S_2)$ if there exist $s_1 \in S_1$ and $s_2 \in S_2$ such that $\pi(s_1)$ and $\pi(s_2)$ belong to the same equivalence class.
Similarly, a plan $\pi$ satisfies the constraint $(\not\sim,S_1,S_2)$ if there exist $s_1 \in S_1$ and $s_2 \in S_2$ such that $\pi(s_1)$ and $\pi(s_2)$  belong to different equivalence classes.
\end{newstuff}
Hence, the constraint $(\sim,s_3,s_5)$ would encode the requirement that the signing and countersigning of the goods received note must be performed by users belonging to the same equivalence class (department, in this example).
\begin{newstuff}
More generally, a constraint of the form $(\sim,s,s')$ represents a weaker constraint than one of the form $(=,s,s')$, since more plans satisfy such a constraint.
Conversely, a constraint of the form $(\nsim,s,s')$ is stronger than $(\ne,s,s')$, as it requires that the two users who perform $s$ and $s'$ are different and, in addition, they belong to different equivalence classes.
\end{newstuff}

\begin{theorem}\label{pro:1level-partition}
For any user set $U$ and any equivalence relation $\sim$ defined on $U$, $\wspi{}(\sim,\not\sim)$ is FPT.
\end{theorem}

\begin{proof}
Consider an instance of the problem $\mathcal{W} = (S,U,\leqslant,A,C)$ and let $V_1,\dots,V_m$ be the equivalence classes of $\sim$.
Then consider the following workflow specification: $\mathcal{W}' = (S,U',\leqslant,A',C')$, where
  \begin{compactitem}
    \item $U' = \set{V_1,\dots,V_m}$;
    \item $A' \subseteq S \times U'$ and $(s,V_i) \in A'$ if there exists $u \in V_i$ such that $(s,u) \in A$;
    \item each constraint of the form $(\sim,S_1,S_2)$ in $C$ is replaced by $(=,S_1,S_2)$ in $C'$; and
    \item each constraint of the form $(\not\sim,S_1,S_2)$ in $C$ is replaced by $(\ne,S_1,S_2)$ in $C'$.
  \end{compactitem}
Observe that $\mathcal{W}$ is satisfiable if and only if $\mathcal{W}'$ is, and deciding the satisfiability of $\mathcal{W}'$ is FPT by Theorem~\ref{thm:1level} and Corollary~\ref{cor:type3-equality-constraints-fpt}.%
\end{proof}

Of course, we could also include counting constraints in the workflow specification.
Let us assume, for ease of explanation, that an equivalence relation partitions a user set into different organizational units.
  \begin{description}
    \item[Separation of duty] The constraint $(1,\set{s',s''})$ requires that users from different organizational units perform $s'$ and $s''$.  More generally, the constraint $(1,\card{S'}-1,S')$ requires that no single unit executes all the steps in $S'$.
    \item[Binding of duty] The constraint $(2,\set{s',s''})$ requires that users from the same organizational unit execute both $s'$ and $s''$.  More generally, the constraint $(\card{S'},S')$ requires that all steps in $S'$ are executed by users from the same unit.
  \end{description}
The other forms of counting constraints introduced in Section~\ref{sec:constraint-types} can be interpreted in analogous ways in the presence of an equivalence relation defined on the set of users.
\end{newstuff}

\section{Organizational Hierarchies}\label{sec:fpt-organizational-constraints}

We now show how we can use multiple equivalence relations to define an organizational hierarchy.
In Section~\ref{subsec:org-hierarchy-fpt}, we describe a fixed-parameter tractable algorithm to solve \wsp\ in the presence of constraints defined over such structures.


Let $S$ be a set.
An {\em $n$-partition} of $S$ is an $n$-tuple $(F_1,\ldots ,F_n)$ such that  $F_1\cup \cdots \cup F_n=S$ and $F_i\cap F_j=\emptyset$ for all $i\neq j\in [n]$.
We will refer to the elements of an $n$-partition as \emph{blocks}.%
\footnote{One or more blocks in an $n$-partition may be the empty set.}

\begin{definition}
Let $(X_1,\ldots ,X_p)$ and $(Y_1,\ldots Y_q)$ be $p$- and $q$-partitions of the same set.
We say that $(Y_1,\ldots Y_q)$ is a {\em refinement} of $(X_1,\ldots ,X_p)$ if for each $i\in [q]$ there exists $j\in [p]$ such that $Y_i\subseteq X_j.$

\end{definition}

\begin{definition}
Let $U$ be the set of users in an organization.
An \emph{organizational $\ell$-hierarchy} is a collection of $\ell$ partitions of $U$, $\xspbr{U}{1},\dots,\xspbr{U}{\ell}$, where  $\xspbr{U}{i}$ is a refinement of $\xspbr{U}{i+1}$.
\end{definition}

The $i$th partition is said to be the $i$th \emph{level} of the hierarchy.
Each member of $\xspbr{U}{i}$ is a subset of $U$; we write $\xspbr{u}{i}$ to denote a block in the $i$th level of the hierarchy.

A constraint of the form $(\sim_i,s_1,s_2)$, for example, is satisfied by plan $\pi$ if \mbox{$\pi(s_1),\pi(s_2) \in \xspbr{u}{i}$} for some \mbox{$\xspbr{u}{i} \in \xspbr{U}{i}$}.
Note, however, that we may still define a constraint \mbox{$(\ne,s_1,s_2)$} which requires that the steps $s_1$ and $s_2$ are performed by different users.

More generally, a constraint of the form $(\sim_i,S_1,S_2)$ is satisfied by plan $\pi$ if there exists $s_1 \in S_1$ and $s_2 \in S_2$ such that $\pi(s_1)$ and $\pi(s_2)$ belong to the same block in $\xspbr{U}{i}$.
A constraint of the form $(\not\sim_i,S_1,S_2)$ is satisfied by $\pi$ if there exist  $s_1 \in S_1$ and $s_2 \in S_2$ such that $\pi(s_1)$ and $\pi(s_2)$ belong to different blocks in $\xspbr{U}{i}$.
Note that if $\pi$ satisfies $(\sim_i,S_1,S_2)$, then it satisfies $(\sim_j,S_1,S_2)$ for all $j > i$.
Conversely, if $\pi$ satisfies $(\not\sim_i,S_1,S_2)$, then it also satisfies $(\not\sim_j,S_1,S_2)$ for all $j < i$.
In other words, for each $S_1,S_2 \subseteq S$, we may and will assume without loss of generality that there is at most one constraint of the form $(\sim_i,S_1,S_2)$ and at most one constraint of the form $(\not\sim_j,S_1,S_2)$.

We now introduce the notion of a canonical hierarchy.
Informally, each level of a canonical hierarchy is different, the top level comprises a single block and the bottom level comprises the set of all singleton blocks.
Two canonical hierarchies are shown in Figure~\ref{fig:canonical-org-hierarchies}, in which $a,\dots,j$ represent users and the rectangles define the partition blocks.
Note that each level is a refinement of the one above.

\begin{figure}[h]\centering
\hspace*{.1\textwidth}
\subfigure[]{\label{fig:org-hrcy-from-man-tree}%
\begin{tikzpicture}[blk/.style={rectangle,inner sep=2pt,draw,minimum height=.6cm,thin},%
                    pt/.style={inner sep=0pt},
                    x=.6cm,y=.6cm,scale=.75,transform shape]
  \node[blk,minimum width=6cm] at (0,6) {};
  \foreach \z in {6,5,4,3,2,1,0}{%
    \foreach \x/\y in {a/-4.5,b/-3.5,c/-2.5,d/-1.5,e/-0.5,f/0.5,g/1.5,h/2.5,i/3.5,j/4.5}
      \node[pt] at (\y,\z) {$\x$};}
  \node[blk,minimum width=5.4cm] at (-0.5,5) {}; \node[blk,minimum width=.6cm] at (4.5,5) {};
  \node[blk,minimum width=2.4cm] at (-3,4) {}; \node[blk,minimum width=3cm] at (1.5,4) {}; \node[blk,minimum width=.6cm] at (4.5,4) {};
  \node[blk,minimum width=2.4cm] at (-3,3) {}; \node[blk,minimum width=2.4cm] at (1,3) {}; \node[blk,minimum width=.6cm] at (3.5,3) {}; \node[blk,minimum width=.6cm] at (4.5,3) {};
  \node[blk,minimum width=2.4cm] at (-3,2) {};
  \node[blk,minimum width=.6cm] at (-0.5,2) {};
  \node[blk,minimum width=1.8cm] at (1.5,2) {};
  \node[blk,minimum width=.6cm] at (3.5,2) {};
  \node[blk,minimum width=.6cm] at (4.5,2) {};
  \node[blk,minimum width=1.8cm] at (-3.5,1) {};
  \node[blk,minimum width=.6cm] at (-1.5,1) {};
  \node[blk,minimum width=.6cm] at (-0.5,1) {};
  \node[blk,minimum width=1.2cm] at (1,1) {};
  \node[blk,minimum width=.6cm] at (2.5,1) {};
  \node[blk,minimum width=.6cm] at (3.5,1) {};
  \node[blk,minimum width=.6cm] at (4.5,1) {};
  \foreach \x in {-4.5,-3.5,-2.5,-1.5,-0.5,0.5,1.5,2.5,3.5,4.5}
    \node[blk,minimum width=.6cm] at (\x,0) {};
\end{tikzpicture}}
\hfill
\subfigure[]{\label{fig:alt-org-hierarchy}%
\begin{tikzpicture}[blk/.style={rectangle,inner sep=2pt,draw,minimum height=.6cm,thin},%
                    pt/.style={inner sep=0pt},
                    x=.6cm,y=.6cm,scale=.75,transform shape]
  \node[blk,minimum width=6cm] at (0,3) {};
  \foreach \z in {3,2,1,0}{%
    \foreach \x/\y in {a/-4.5,b/-3.5,c/-2.5,d/-1.5,e/-0.5,f/0.5,g/1.5,h/2.5,i/3.5,j/4.5}
      \node[pt] at (\y,\z) {$\x$};}
  \node[blk,minimum width=2.4cm] at (-3,2) {}; \node[blk,minimum width=3cm] at (1.5,2) {}; \node[blk,minimum width=.6cm] at (4.5,2) {};
  \node[blk,minimum width=2.4cm] at (-3,1) {};
  \node[blk,minimum width=.6cm] at (-0.5,1) {};
  \node[blk,minimum width=1.8cm] at (1.5,1) {};
  \node[blk,minimum width=.6cm] at (3.5,1) {};
  \node[blk,minimum width=.6cm] at (4.5,1) {};
  \foreach \x in {-4.5,-3.5,-2.5,-1.5,-0.5,0.5,1.5,2.5,3.5,4.5}
    \node[blk,minimum width=.6cm] at (\x,0) {};
\end{tikzpicture}}
\hspace*{.1\textwidth}
\caption{Two canonical organizational hierarchies}\label{fig:canonical-org-hierarchies}
\end{figure}
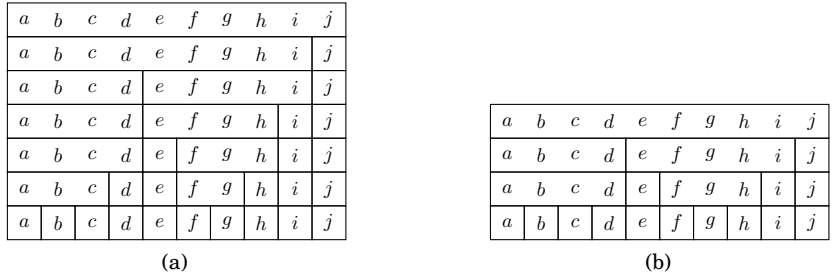

More formally, we have the following definition.

\begin{definition}
Let $\mathcal{H} = {\xspbr{U}{1},\dots,\xspbr{U}{\ell}}$, where $\xspbr{U}{i}$ is a refinement of $\xspbr{U}{i+1}$,  be a hierarchy.
We say $\mathcal{H}$ is \emph{canonical} if it satisfies the following conditions:
  \begin{inparaenum}[(i)]
    \item $\xspbr{U}{i} \ne \xspbr{U}{i+1}$;
    \item $\xspbr{U}{\ell}$ is a $1$-partition containing the set $U$;
    \item $\xspbr{U}{1}$ is an $n$-partition containing every singleton set (from $U$).
  \end{inparaenum}
\end{definition}

Let $\xspbr{U}{1},\dots,\xspbr{U}{\ell}$ be some hierarchy and let $C$ be a set of workflow constraints.
We conclude this section by showing how we may convert the hierarchy into a canonical hierarchy by first removing duplicate levels, adding suitable top and bottom levels (if required), and making appropriate adjustments to $C$.
More formally, we perform the following operations:
  \begin{compactitem}
    \item If $\xspbr{U}{i} = \xspbr{U}{i+1}$ for some $i$ then we replace all constraints of the form $(\sim_{i+1},S_1,S_2)$ and \mbox{$(\not\sim_{i+1},S_1,S_2)$} with constraints of the form $(\sim_i,S_1,S_2)$ and \mbox{$(\not\sim_i,S_1,S_2)$}, respectively.
          We then remove $\xspbr{U}{i+1}$ from the hierarchy as there are now no constraints that apply to $\xspbr{U}{i+1}$.
    \item If no partition in the hierarchy has one element (consisting of a single block $U$), then add such a partition to the hierarchy.
          Clearly every partition is a refinement of the $1$-partition $(U)$.
    \item If no partition in the hierarchy has $n$ elements, then add such a partition to the hierarchy.
          Clearly such a partition is a refinement of every other partition.
    \item Finally, we renumber the levels and the constraints where appropriate with consecutive integers.
  \end{compactitem}

The conversion of a hierarchy to canonical form can be performed in $O(\ell n+c)$ time (since we require $O(\ell n)$ time to find all layers that may be deleted and then delete them, and $O(c)$ time to update the constraints).
The number of levels in the resulting canonical hierarchy is no greater than $\ell + 2$.

\subsection{Organizational Hierarchies from Management Structures}

\newcommand{\ule}{\prec}
\newcommand{\uge}{\succ}
\newcommand{\uleq}{\preccurlyeq}
\newcommand{\ugeq}{\succcurlyeq}

We now illustrate how organization hierarchies may be constructed in a systematic fashion from management structures.
Given a set of users $U$, we assume that an organization defines a hierarchical binary relation $\ugeq$ on $U$ in order to specify management responsibilities and reporting lines.
We assume that the Hasse diagram of $(U,\ugeq)$ is a directed tree in which non-leaf nodes represent users with some managerial responsibility and edges are directed from root node to leaf nodes.
Let $G_{\rm man} = (U,E_{\rm man})$ denote the Hasse diagram of $(U,\ugeq)$.
The fact that $G_{\rm man}$ is a tree means that no user has more than one manager.
A user $u$ has direct responsibility for (or is the line manager of) user $v$ if $(u,v) \in E_{\rm man}$.
We also assume that the out-degree of a non-leaf node is at least two.

We now describe one method by which an organizational hierarchy may be derived from a management tree.
Given a management tree $G_{\rm man}$ we iteratively construct management trees with fewer and fewer nodes as follows:
  \begin{compactenum}[(1)]
    \item we first identify every sub-tree in which there is a single non-leaf node;
    \item for each such sub-tree we form a single leaf node whose label is formed from the labels for the respective leaf nodes;
    \item for each resulting sub-tree we form a single node whose label is formed from the labels of the child and parent nodes.
  \end{compactenum}
We then repeat for the resulting tree, terminating when we have a tree containing a single node.

The above procedure is illustrated in Figure~\ref{fig:man-tree-to-org-hierarchy}.
The figure shows a sequence of trees, the first of which defines the management tree in which each node is labeled with a single user.
Each management tree thus derived is associated with a partition; the corresponding partition of $U$ is written below each tree in Figure~\ref{fig:man-tree-to-org-hierarchy}, with a vertical bar indicating the block boundaries.
By construction, the collection of partitions forms a canonical organizational hierarchy.
The organizational hierarchy derived from the management tree in Figure~\ref{fig:man-tree-to-org-hierarchy} is displayed in Figure~\ref{fig:org-hrcy-from-man-tree}.
Note that the number of levels in the organizational hierarchy is equal to $2p + 1$, where $p$ is the number of edges in the longest directed path in $G_{\rm man}$.

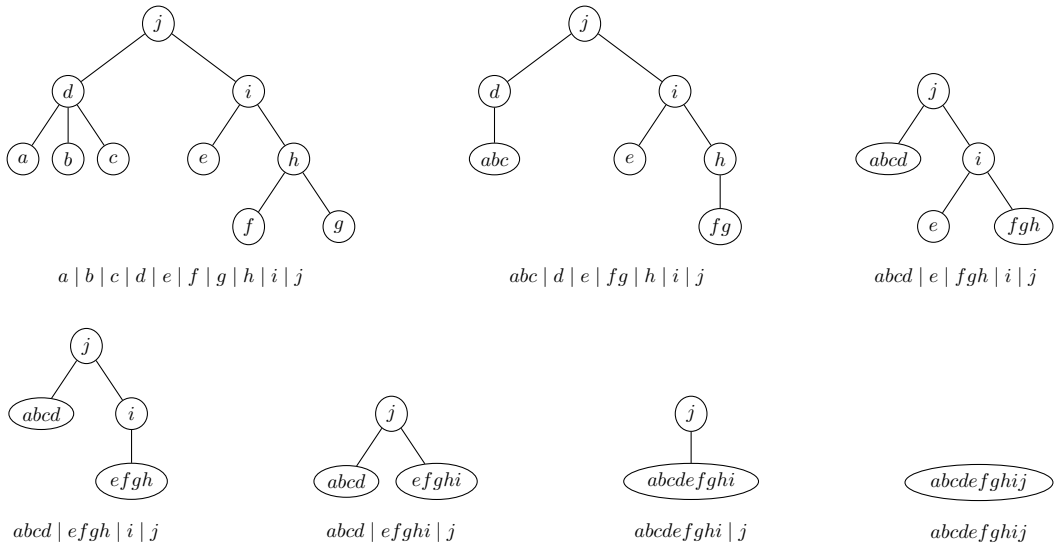
\begin{figure}[!ht]\centering
\begin{tikzpicture}[usr/.style={ellipse,inner sep=2pt,draw,thin,minimum width=16pt,minimum height=16pt},x=0.8cm,y=0.6cm,scale=.75,transform shape]
  \node[usr] (j) at (0,6) {$j$};
  \node[usr] (d) at (-2,4) {$d$};
  \node[usr] (i) at (2,4) {$i$};
  \node[usr] (a) at (-3,2) {$a$};
  \node[usr] (b) at (-2,2) {$b$};
  \node[usr] (c) at (-1,2) {$c$};
  \node[usr] (e) at (1,2) {$e$};
  \node[usr] (h) at (3,2) {$h$};
  \node[usr] (f) at (2,0) {$f$};
  \node[usr] (g) at (4,0) {$g$};
  \draw (j) -- (d); \draw (j) -- (i); \draw (d) -- (a); \draw (d) -- (b); \draw (d) -- (c); \draw (i) -- (e); \draw (i) -- (h); \draw (h) -- (f); \draw (h) -- (g);
  \node at (0.5,-1.5) {$a \mid b \mid c \mid d \mid e \mid f \mid g \mid h \mid i \mid j$};
\end{tikzpicture}
\hfill%
\begin{tikzpicture}[usr/.style={ellipse,inner sep=2pt,draw,thin,minimum width=16pt,minimum height=16pt},scale=.75,x=0.8cm,y=0.6cm,transform shape]
  \node[usr] (j) at (0,6) {$j$};
  \node[usr] (d) at (-2,4) {$d$};
  \node[usr] (i) at (2,4) {$i$};
  \node[usr] (abc) at (-2,2) {$abc$};
  \node[usr] (e) at (1,2) {$e$};
  \node[usr] (h) at (3,2) {$h$};
  \node[usr] (fg) at (3,0) {$fg$};
  \draw (j) -- (d); \draw (j) -- (i); \draw (d) -- (abc); \draw (i) -- (e); \draw (i) -- (h); \draw (h) -- (fg);
  \node at (0.5,-1.5) {$a b c \mid d \mid e \mid f g \mid h \mid i \mid j$};
\end{tikzpicture}
\hfill%
\begin{tikzpicture}[usr/.style={ellipse,inner sep=2pt,draw,thin,minimum width=16pt,minimum height=16pt},scale=.75,x=0.8cm,y=0.6cm,transform shape]
  \node[usr] (j) at (0,6) {$j$};
  \node[usr] (abcd) at (-1,4) {$abcd$};
  \node[usr] (i) at (1,4) {$i$};
  \node[usr] (e) at (0,2) {$e$};
  \node[usr] (fgh) at (2,2) {$fgh$};
  \draw (j) -- (abcd); \draw (j) -- (i); \draw (i) -- (e); \draw (i) -- (fgh);
  \node at (0.5,0.5) {$a b c d \mid e \mid f g h \mid i \mid j$};
\end{tikzpicture}
\hfill\\[12pt]%
\begin{tikzpicture}[usr/.style={ellipse,inner sep=2pt,draw,thin,minimum width=16pt,minimum height=16pt},scale=.75,x=0.8cm,y=0.6cm,transform shape]
  \node[usr] (j) at (0,6) {$j$};
  \node[usr] (abcd) at (-1,4) {$abcd$};
  \node[usr] (i) at (1,4) {$i$};
  \node[usr] (efgh) at (1,2) {$efgh$};
  \draw (j) -- (abcd); \draw (j) -- (i); \draw (i) -- (efgh);
  \node at (0,.5) {$a b c d \mid e f g h \mid i \mid j$};
\end{tikzpicture}
\hfill%
\begin{tikzpicture}[usr/.style={ellipse,inner sep=2pt,draw,thin,minimum width=16pt,minimum height=16pt},scale=.75,x=0.8cm,x=0.8cm,y=0.6cm,transform shape]
  \node[usr] (j) at (0,6) {$j$};
  \node[usr] (abcd) at (-1,4) {$abcd$};
  \node[usr] (efghi) at (1,4) {$efghi$};
  \draw (j) -- (abcd); \draw (j) -- (efghi);
  \node at (0,2.5) {$a b c d \mid e f g h i \mid j$};
\end{tikzpicture}
\hfill%
\begin{tikzpicture}[usr/.style={ellipse,inner sep=2pt,draw,thin,minimum width=16pt,minimum height=16pt},scale=.75,x=0.8cm,y=0.6cm,transform shape]
  \node[usr] (j) at (0,6) {$j$};
  \node[usr] (abcdefghi) at (0,4) {$abcdefghi$};
  \draw (j) -- (abcdefghi);
  \node at (0,2.5) {$a b c d e f g h i \mid j$};
\end{tikzpicture}
\hfill%
\begin{tikzpicture}[usr/.style={ellipse,inner sep=2pt,draw,thin,minimum width=16pt,minimum height=16pt},scale=.75,x=0.8cm,y=0.6cm,transform shape]
  \node[usr] (j) at (0,6) {$abcdefghij$};
  \node at (0,4.5) {$a b c d e f g h i j$};
\end{tikzpicture}
\caption{Building the blocks of an organizational hierarchy from a management tree}\label{fig:man-tree-to-org-hierarchy}
\end{figure}

Having constructed the organizational hierarchy, we may now define constraints on step execution.
We will use our purchase order workflow from Figure~\ref{fig:example-workflow} as an example and the organizational hierarchy in Figure~\ref{fig:org-hrcy-from-man-tree}.

We could, for example, define the constraint $(\sim_5,s_1,s_2)$.
In the absence of other constraints, this constraint means that users from the set $\set{a,b,c,d}$ or $\set{e,f,g,h,i}$ (which we might suppose represent two distinct departments within the management structure) or user $j$ could raise (step $s_1$) and approve (step $s_2$) purchase orders, but an attempt by a user from one department to approve an order raised by a member of another department would violate the constraint.

We could define a second constraint $(\nsim_4,s_1,s_2)$, which means that user $i$ must perform one of $s_1$ and $s_2$ (and also means that no user from $\set{a,b,c,d,j}$ can perform either $s_1$ or $s_2$ because there would be no way to simultaneously satisfy constraints $(\sim_5,s_1,s_2)$ and $(\nsim_4,s_1,s_2)$).
If we assume that junior members of the department (users $e$, $f$, $g$ and $h$) are not authorized to approve purchase orders, the collective effect of the two constraints above and the authorization policy is to require that
  \begin{inparaenum}[(a)]
    \item purchase orders are only approved by managers, and
    \item purchase orders are only raised by junior members of staff.
  \end{inparaenum}

Pursuing the last point briefly, it has long been recognized that a limitation of role-based access control is the ``feature'' that (senior) users assigned to the most powerful roles accrue all the permissions of more junior roles (see~\citeN{MoLu99}, for example).
It is interesting to note that the constraints and the method of constructing an organizational hierarchy described above can be used to restrict the steps that senior managers can perform.

In summary, we believe that our definition of organizational hierarchy provides an appropriate way of modeling hierarchical management structures and supports the specification of constraints that provide greater flexibility than those in the literature~\cite{BeFeAt99,cram:sacmat05,WangLi10}, which have focused on constraints involving only $=$ and $\ne$.
Moreover, as we will see in the next section, the complexity of \wsp\ for these new constraints remains fixed-parameter tractable.

Finally, we note that there are several ways in which the construction of an organizational hierarchy from a management tree described above could be modified.
At each iteration we could, for example, collapse the root node and all the leaf nodes into a single node.
In doing so, we remove the distinction between the line manager of an organizational unit and the remaining members of the unit.
If we adopt this approach for the management tree in Figure~\ref{fig:man-tree-to-org-hierarchy}, we derive the organizational hierarchy shown in Figure~\ref{fig:alt-org-hierarchy}.
Clearly this construction results in fewer layers in the organizational hierarchy (equal to $p+1$, where $p$ is the length of the longest directed path in the management tree) and, therefore, supports fewer choices of workflow constraints.

Each method will give rise to different organizational hierarchies, some with more levels, some with fewer, with each hierarchy allowing for the specification of a different set of constraints.
The method used to construct an organizational hierarchy will usually depend on the workflow, the organization and the type of constraints that are required.
\begin{newstuff}
An alternative approach to both those described above would be to ``stratify'' the management tree into levels and, working from the bottom level up, collapse all departments at a specific level into single nodes.
Using the management tree in Figure~\ref{fig:man-tree-to-org-hierarchy}, for example, the users $f$ and $g$ form the lowest level in the stratified tree and would be merged into a single unit first; this would be followed by the merging of users $a$, $b$ and $c$ and of $e$, $f$, $g$ and $h$.
The resulting canonical hierarchy will be rather similar to the one depicted in Figure~\ref{fig:org-hrcy-from-man-tree}, although the departments will form at different levels in the new hierarchy.
\end{newstuff}
The study of such hierarchies and the utility of the constraints that can be defined over them will be the subject of future work.

\subsection{Organizational Hierarchy Constraints}\label{subsec:org-hierarchy-fpt}

We have seen that if we are given a single equivalence relation and only use the binary relations $\sim$ and $\not\sim$ then $\wspi{2}(\sim,\not\sim)$ may be transformed into an instance of \mbox{$\wspi{2}(=,\ne)$}, which is known to be FPT.
We prove in Theorem~\ref{thm:llevel} that the problem remains in FPT for organizational hierarchies with $\ell$ levels (defined by $\ell$ equivalence relations).
In fact, the results in Theorem~\ref{thm:1level} and Proposition~\ref{pro:1level-partition} correspond to special cases of Theorem~\ref{thm:llevel}, in which the hierarchy has two levels.
Figure~\ref{fig:two-level-hierarchies} illustrates these hierarchies, where each user is represented by an unfilled circle and blocks of users are enclosed by a rectangle.
Conversely, it is these special cases that provide the foundation for the bottom-up iterative method that we use in the proof of Theorem~\ref{thm:llevel} to solve \wsp\ for more complex hierarchical structures.

\begin{figure}[h]
  \subfigure[Single users]{
  \begin{minipage}{.45\textwidth}\centering
    \begin{tikzpicture}[rct/.style={rectangle,draw,fill=black!20,minimum width=0.4cm,minimum height=0.5cm},%
                        crc/.style={circle,draw,inner sep=0pt,fill=white,minimum width=2mm}]
      \node[rct,minimum width=3.5cm] (rtop) at (0,1) {};
      \foreach \x/\y in {-1.5/0,-1/1,-0.5/2,0/3,0.5/4,1/5,1.5/6}
        \node[rct] (r\y) at (\x,0) {};
      \foreach \x in {0,1}{
        \foreach \y in {-1.5,-1,-0.5,0,0.5,1,1.5}
          \node[crc] at (\y,\x) {};
      }
      \foreach \x in {0,1,2,3,4,5,6}
        \draw (rtop) -- (r\x.north);
    \end{tikzpicture}
    \vspace*{6pt}
    \end{minipage}
  }
  \hfill
  \subfigure[Non-trivial equivalence relation]{
    \begin{minipage}{.5\textwidth}\centering
    \begin{tikzpicture}[rct/.style={rectangle,draw,fill=black!20,minimum width=0.4cm,minimum height=0.5cm},%
                        crc/.style={circle,draw,inner sep=0pt,fill=white,minimum width=2mm}]
      \node[rct,minimum width=3.5cm] (rtop) at (0,1) {};
      \foreach \x/\y/\z in {-1.5/0/0.4cm,-0.75/1/0.9cm,0.75/2/1.9cm}
        \node[rct,minimum width=\z] (r\y) at (\x,0) {};
      \foreach \x in {0,1}{
        \foreach \y in {-1.5,-1,-0.5,0,0.5,1,1.5}
          \node[crc] at (\y,\x) {};
      }
      \foreach \x in {0,1,2}
        \draw (rtop) -- (r\x.north);
    \end{tikzpicture}
    \vspace*{6pt}
    \end{minipage}
  }
\caption{Two-level hierarchies}\label{fig:two-level-hierarchies}
\end{figure}
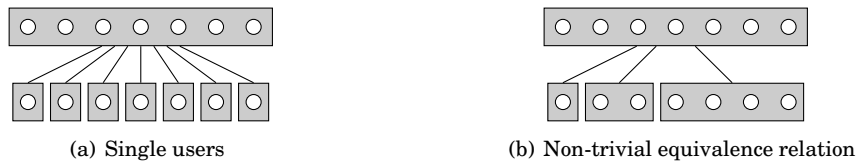

Recall Wang and Li proved that \wsp\ is not FPT, in general.
One crucial factor in determining the complexity of \wsp\ is the nature of the binary relations used to define entailment constraints.
Informally, Wang and Li showed that for a particular choice of relational structure on the user set, \wsp\ is an instance of {\sc Independent Set}, which is known to be W[1]-complete.
Constraints based on equivalence relations, however, do not compromise the fixed parameter tractability of \wsp\ because of the particular structure that is imposed on the user set---namely, a partition into no more than $2^k$ blocks.

Before proving the main result of this section, we consider canonical hierarchies with exactly three levels.
There are several reasons for doing so:
  \begin{itemize}
    \item if we are given a non-trivial equivalence relation $\sim$ and we are interested in \mbox{\wsp$(=,\ne,\sim,\nsim)$} then there are three levels in the organizational hierarchy;
    \item constraints containing $\sim$ and $\nsim$ have useful applications for many types of authorization policies; and
    \item three-level hierarchies represent the ``tipping point'' at which \wsp\ becomes hard, in the sense that no polynomial kernel exists (see Section~\ref{sec:kernelization}).
  \end{itemize}

There are several situations in which we may have a single non-trivial equivalence relation.
Perhaps the most obvious one arises when a set of users is grouped into distinct departments or organizational units, as we have previously noted.
Other possibilities arise from a natural re-interpretation of the authorization relation $A \subseteq S \times U$: specifically, we define $u \sim_A u'$ if and only if $u$ and $u'$ are authorized for the same workflow steps.
Then there are a maximum of $2^k$ equivalence classes (each associated with a particular subset of workflow steps).
In a role-based view of authorization~\cite{SaCoFeYo96}, a set of permissions (such as execution of workflow steps) defines a role.
With this interpretation, a constraint of the form $(\sim,s_1,s_2)$ requires that $s_1$ and $s_2$ are performed by users that are assigned to the same role(s), with an analogous interpretation for $(\nsim,s_1,s_2)$.\footnote{Of course, we could replace $A$ with user- and permission-role assignment relations, but we could still derive the same equivalence classes.}

We may also consider an authorization policy that associates users and workflow steps with a security label, as in the Bell-LaPadula security model~\cite{bell:secu76}.
More formally, let $(L,\leqslant)$ be a partially ordered set of security labels and $\lambda : U \cup S \rightarrow L$ a function that associates each user and step with a security label.
Then a user is authorized to perform step $s$ if and only if $(s,u) \in A$ and $\lambda(u) \geqslant \lambda(s)$.
Clearly $\sim_\lambda$, where $u \sim_\lambda u'$ if and only if $\lambda(u) = \lambda(u')$ is an equivalence relation.
The constraint $(\sim,s_1,s_2)$ requires that steps $s_1$ and $s_2$ be performed by users with the same security clearance.
In short, there seem to be a number of situations in which the use of constraints defined by equivalence relations will be useful.%

\begin{theorem}\label{thm:llevel}
Given a workflow $(S,U,\leqslant,A,C)$ and a canonical hierarchy with $\ell$ levels, $\wspi{2}(\sim_1,\not\sim_1,\dots,\sim_{\ell},\not\sim_\ell)$ can be solved in time \mbox{$\widetilde{O}(3^k n(c + n))$}, where $n$, $k$ and $c$ are the numbers of users, steps and constraints, respectively.
\end{theorem}

Theorem~\ref{thm:1level} is, essentially, a special case of the above result, in which the canonical hierarchy contains two levels, where $\xspbr{U}{1} = (\set{u_1},\dots,\set{u_n})$ and $\xspbr{U}{2} = (U)$.
\begin{newstuff}
To prove Theorem~\ref{thm:llevel}, we identify particular types of blocks in the hierarchy (those shaded in Figure~\ref{fig:hierarchy-canonical}) and solve multiple instances of \wsp\ for each of those ``significant'' blocks.
The results for significant blocks at a particular level are then used to solve instances of \wsp\ for significant blocks at higher levels in the hierarchy.
\end{newstuff}

\begin{figure}[h]\centering
\begin{tikzpicture}[blk/.style={rectangle,inner sep=2pt,draw,minimum height=.6cm,thin},%
                    sblk/.style={rectangle,fill=black!20,inner sep=2pt,draw,minimum height=.6cm,thin},
                    pt/.style={inner sep=0pt},
                    x=.6cm,y=.6cm,scale=.75,transform shape]
  \node[sblk,minimum width=6cm] at (0,6) {};
  \node[sblk,minimum width=5.4cm] at (-0.5,5) {};
  \node[sblk,minimum width=.6cm] at (4.5,5) {};
  \node[sblk,minimum width=2.4cm] at (-3,4) {}; \node[sblk,minimum width=3cm] at (1.5,4) {}; \node[blk,minimum width=.6cm] at (4.5,4) {};
  \node[blk,minimum width=2.4cm] at (-3,3) {}; \node[sblk,minimum width=2.4cm] at (1,3) {}; \node[sblk,minimum width=.6cm] at (3.5,3) {}; \node[blk,minimum width=.6cm] at (4.5,3) {};
  \node[blk,minimum width=2.4cm] at (-3,2) {};
  \node[sblk,minimum width=.6cm] at (-0.5,2) {};
  \node[sblk,minimum width=1.8cm] at (1.5,2) {};
  \node[blk,minimum width=.6cm] at (3.5,2) {};
  \node[blk,minimum width=.6cm] at (4.5,2) {};
  \node[sblk,minimum width=1.8cm] at (-3.5,1) {};
  \node[sblk,minimum width=.6cm] at (-1.5,1) {};
  \node[blk,minimum width=.6cm] at (-0.5,1) {};
  \node[sblk,minimum width=1.2cm] at (1,1) {};
  \node[sblk,minimum width=.6cm] at (2.5,1) {};
  \node[blk,minimum width=.6cm] at (3.5,1) {};
  \node[blk,minimum width=.6cm] at (4.5,1) {};
  \foreach \x in {-4.5,-3.5,-2.5,0.5,1.5}
    \node[sblk,minimum width=.6cm] at (\x,0) {};
  \foreach \x in {-1.5,-0.5,2.5,3.5,4.5}
    \node[blk,minimum width=.6cm] at (\x,0) {};
  \foreach \z in {6,5,4,3,2,1,0}{%
    \foreach \x/\y in {a/-4.5,b/-3.5,c/-2.5,d/-1.5,e/-0.5,f/0.5,g/1.5,h/2.5,i/3.5,j/4.5}
      \node[pt] at (\y,\z) {$\x$};}
\end{tikzpicture}
  \caption{The canonical hierarchy of Figure~\ref{fig:man-tree-to-org-hierarchy} with its significant blocks shaded}\label{fig:hierarchy-canonical}
\end{figure}

\begin{proof}[of Theorem~\ref{thm:llevel}]
Each level in a canonical hierarchy is a refinement of the one below it and no two levels are equal, so we have
  $
    n = |\xspbr{U}{1}| > \dots > |\xspbr{U}{\ell}| = 1,
  $
and we may conclude that $\ell \leqslant n$.

We say $V \in \xspbr{U}{i}$  is \emph{significant} if $V \not\in \xspbr{U}{i+1}$.
We define the \emph{level range} of $V$ to be an interval $[a,b]$, where $a$ is the least value $i$ such that $V \in \xspbr{U}{i}$ and $b$ is the largest value $i$ such that $V \in \xspbr{U}{i}$.
The level range of block $\set{a,b,c,d}$ in Figure~\ref{fig:hierarchy-canonical} is $[3,5]$, for example.

Each significant block $V$ with level range $[a,b]$, $a > 1$, can be partitioned into blocks in level $(a-1)$.
We denote this set of blocks by $\children{V}$.
Each significant block $V$ with level range $[1,b]$ comprises a single user (see Figure~\ref{fig:hierarchy-canonical}).
It is easy to see that the graph $G = (\mathcal{V},E)$, where $\cal V$ is the set of significant blocks and $(V_1,V_2) \in E$ if $V_1 \in \children{V_2}$, is a tree\begin{newstuff}, in which the leaf nodes are blocks with level range $[1,b]$ for some $b < \ell$\end{newstuff}.

\begin{newstuff}
Given an instance $\cal I$ of $\wspi{2}(\sim_1,\not\sim_1,\dots,\sim_{\ell},\not\sim_\ell)$, every subset $F$ of $S$ and every significant block $V$ with children $\Delta(V)$ defines an instance of \wsp\ in which:
  \begin{compactitem}
	\item the set of steps is $F$;
	\item the set of users is $\Delta(V)$;
	\item the authorization relation $A'$ is a subset of $F \times \Delta(V)$, where $(s,W) \in A'$ if and only if there exists a user in $v \in W$ such that $(s,v) \in A$;
	\item the set of constraints comprises those constraints in $C$ of the form $(\rho,S_1,S_2)$, where $\rho$ is $\sim_i$ or $\not\sim_i$ with $a \leqslant i \leqslant b$.
  \end{compactitem}
We denote this derived instance of \wsp\ by ${\cal I}_{F,V}$.
Note that if $V$ has level range $[1,b]$, then ${\cal I}_{F,V}$ asks whether a single user is authorized to perform all the steps in
$F$ without violating any constraints defined between levels $1$ and $b$ of the hierarchy.
If $V$ has level range $[a,b]$, with $a > 1$, then ${\cal I}_{F,V}$ is solved using the approach similar to that described in the proof of Theorem~\ref{thm:master}.
When building the matrix, the entry indexed by $G \subseteq F$ and $W$ is defined to be $0$ if and only if $G \ne \emptyset$ is ineligible or ${\cal I}_{G,W}$ is a no-instance of WSP.
Thus, a non-zero matrix entry indicates the steps in $F$ could be assigned to the block $W$ (meaning that no constraints in levels $1,\dots,a-1$ would be violated) and that no constraints would be violated in levels $a,\dots,b$ by allocating a single block to $F$.
Hence, we can solve ${\cal I}_{F,V}$ if we can solve ${\cal I}_{F,W}$ for all $W \in \Delta(V)$.

Note, finally, that $U$ is a significant set and a solution for ${\cal I}_{S,U}$ is a solution for $\cal I$.
Thus our algorithm for solving $\cal I$ solves ${\cal I}_{F,V}$ for all significant sets $V$ with level range $[a,b]$ from $a = 1$ to $a = \ell$ and all subsets $F$ of $S$.
\end{newstuff}

We now consider the complexity of this algorithm.
Consider the significant block $V$ with $m$ children.
\begin{newstuff}
If $m = 0$ then $V = \set{u}$ for some $u \in U$ and solving ${\cal I}_{F,V}$ amounts to identifying whether $F$ is an eligible set and whether $u$ is authorized for all steps in $F$.
For fixed $V$ (with $m = 0$), solving ${\cal I}_{F,V}$ for all $F \subseteq S$ takes time $O(2^kc)$.
There are exactly $n$ significant sets, one per user, with no children.
If $m > 0$ then the time taken to solve ${\cal I}_{F,V}$ is $\widetilde{O}(2^{\card{F}}(c + m^2))$, by Theorem~\ref{thm:1level}.
Hence the time taken to solve ${\cal I}_{F,V}$ for all $F \subseteq S$ (for fixed $V$) is $\widetilde{O}(3^k(c + m^2))$.
As we observed earlier, the set of significant blocks ordered by subset inclusion forms a tree.
Moreover, every non-leaf node in $G$ has at least two children, which implies that $G$ has no more than $2n-1$ nodes (so  $\card{\mathcal{V}} \leqslant 2n-1$), so there are at most $n-1$ significant sets with $2$ or more children.

The total time taken, therefore, is
  \[
	O(2^k cn) + \sum_{V \in \cal V} \widetilde{O}(3^k(c + m^2_V) = \widetilde{O}(3^k cn) + \sum_{V \in \cal V} \widetilde{O}(3^k m^2_V),
  \]
where $m_V$ denotes the number of children of $V$.

Now for some $b \geqslant 0$, we have
  \[
    \sum_{V \in \cal V} \widetilde{O}(m^2_V) = \sum_{V \in \cal V} O((m_V \log^b m_V) ^2) 
											 \leqslant \max_{V \in \cal V} \log^{2b} m_V \sum_{V \in \cal V} O(m^2_V) 
											 = O(n^2 \log^{2b} n) = \widetilde{O}(n^2).
  \]
Hence, we conclude that the total time taken to compute $\phi_V$ for all $V$ is $\widetilde{O}(3^kcn + 3^k n^2)) = \widetilde{O}(3^k n(c+n))$.
\end{newstuff}
\end{proof}

\begin{newstuff}
\begin{remark}\label{rem:optimized-proof}
The algorithm in the above proof can be optimized by computing a single matrix for each significant set $V$ (with rows indexed by $\Delta(V)$ and columns indexed by subsets of $S$), which can be used to solve ${\cal I}_{F,V}$ for all $F \subseteq S$.
This matrix can be built in time $O(cm2^k)$ and the solution to ${\cal I}_{F,V}$, for $F \subseteq S$, can be computed in time $\widetilde{O}(2^{\card{F}} m^2)$.
Hence, the optimized algorithm runs in time $\widetilde{O}(cm 2^k + m^2 3^k)$, for fixed $V$, and in time $\widetilde{O}(cn 2^k + n^2 3^k)$ overall.
\end{remark}
\end{newstuff}

\begin{theorem}\label{thm:type3}
Let $\sim_1,\dots,\sim_\ell$ define a canonical organizational hierarchy.
Let \mbox{$W = (S,U,\leqslant,A,C \cup C_{\sim} \cup C_{\not\sim})$} be a workflow, where $C$ is the set of Type $2$ constraints, $C_{\sim}$ is the set of Type $3$ constraints of the form $(\sim_i,S_1,S_2)$ and $C_{\not\sim}$ is the set of Type $3$ constraints  of the form $(\not\sim_i,S_1,S_2)$.
Then the satisfiability of $W$ can be determined in time
  \[
    \widetilde{O}( (c+2c')n 2^{k+c'} + n^2 3^{k+c'} ),
  \]
where $c = \card{C} + \card{C_{\not\sim}}$ and $c' = \card{C_{\sim}}$.
Moreover, $c' \leqslant 3^k$, so $\wspi{3}(\sim_1,\dots,\sim_\ell,\not\sim_1,\dots,\not\sim_\ell)$ is FPT.
\end{theorem}
The proof of this result can be found in the appendix.

\section{Kernelization}\label{sec:kernelization}

\newcommand{\np}{\textrm{NP}}
\newcommand{\conp}{\textrm{coNP}}
\newcommand{\poly}{\textrm{poly}}

Formally, a \emph{parameterized problem} $P$ can be represented as a relation $P\subseteq \Sigma^* \times \mathbb{N}$ over a finite alphabet $\Sigma$.
The second component is call the {\em parameter} of the problem.
In particular, {\sc WSP} is a parameterized problem with parameter $k$, the number of steps.
\begin{newstuff}
We denote the size of a problem instance $(x,k)$ by $\card{x} + k$.
In this section, we are interested in transforming an instance of \wsp\ into a new instance of \wsp\ whose size is dependent only on $k$.
This type of transformation is captured in the following definition.
\end{newstuff}

\begin{definition}
Given a parameterized problem $P$, a \emph{kernelization of $P$} is an algorithm that maps an instance $(x,k)$ to an instance $(x',k')$ in time polynomial in $|x|+k$ such that
  \begin{inparaenum}[(i)]
    \item $(x,k)\in P$ if and only if $(x',k')\in P$, and
    \item $k'+|x'|\leqslant g(k)$ for some function  $g$;
  \end{inparaenum}
$(x',k')$ is the \emph{kernel} and $g$ is the {\em size} of the kernel.
\end{definition}

\begin{newstuff}
Note that a kernelization provides a form of preprocessing aimed at compressing the given instance of the problem.
The compressed instance can be solved using \emph{any} suitable algorithm (such as a SAT solver), not necessarily by an FPT algorithm.
\end{newstuff}
It is well-known and easy to prove that a decidable parameterized problem is FPT if and only if it has a kernel~\cite{FlumGrohe06}.
If $g(k)=k^{O(1)}$, then we say $(x',k')$ is a {\em polynomial-size} kernel.

Polynomial-size kernels are particularly useful in practice as they often allow us to reduce the size of the input of the problem under consideration to an equivalent problem with an input of significantly smaller size. This preprocessing often allows us to solve the original problem more quickly.
Unfortunately, many fixed-parameter tractable problems have no polynomial-size kernels (unless $\conp \subseteq \np/\poly$, which is highly unlikely~\cite{BDFH09,BodJanKra,BTY09,DLS09}).

\begin{newstuff}
In order to illustrate the benefits of kernelization, we first state and prove three simple results, the first two of which extend a result of \citeN{FeFrHeNaRo11}.
We then show that $\wspi{1}(=,\ne)$ has a kernel with at most $k$ users.

\begin{proposition}\label{pro:wsp-ne-kernel}
$\wspi{}(\ne)$ has a kernel with at most $k(k-1)$ users.
Moreover, a kernel with at most $k(k-1)$ users exists if we extend the set of constraints to include counting constraints of the form $(1,t,S')$.
\end{proposition}

\begin{proof}
Let $W = (S,U,\leqslant,A,C)$ be a workflow in which all constraints have the form $(\ne,S_1,S_2)$.
Let $S_{\rm easy}$ be the set of steps such that each step has at least $k$ authorized users and let $S_{\rm hard}$ be $S \setminus S_{\rm easy}$.
Now consider the workflow $W_{\rm hard} = (S_{\rm hard},U_{\rm hard},\leqslant,A_{\rm hard},C_{\rm hard})$, where $u \in U_{\rm hard}$ if and only if $u$ is authorized for at least one step in $S_{\rm hard}$, $A_{\rm hard} = (S_{\rm hard} \times U) \cap A$, and $(\ne,S_1,S_2) \in C_{\rm hard}$ if and only if $(\ne,S_1,S_2) \in C$ and $S_1,S_2 \subseteq S_{\rm hard}$.
A counting constraint of the form $(1,t_r,S')$ is replaced by the counting constraint $(1,t_r,S' \setminus S_{\rm easy})$.

We now solve the \wsp\ instance defined by $W_{\rm hard}$ and show that this allows us to compute a solution for $W$.
If $W_{\rm hard}$ is a no instance, then $W$ cannot be satisfiable either (since $C_{\rm hard} \subseteq C$).
Conversely, if it is a yes instance, then there exists a plan $\pi_{\rm hard} : S_{\rm hard} \rightarrow U_{\rm hard}$.
Moreover, we can extend $\pi_{\rm hard}$ to a plan $\pi : S \rightarrow U$, so $W$ is satisfiable.
Specifically, we allocate a different user from $U \setminus \pi_{\rm hard}(S_{\rm hard})$ to each step in $s \in S_{\rm easy}$ (which is possible since there are at least $k$ users authorized to perform $s$ and only $k$ steps in total) and define $\pi(s) = \pi_{\rm hard}(s)$ for all $s \in S_{\rm hard}$.
Clearly, $\pi$ does not violate any constraint of the form $(\ne,S_1,S_2)$ or $(1,t,S')$.%
\footnote{Note that this is not true for counting constraints of the form $(t_\ell,t_r,S')$ when $t_\ell > 1$.}

In other words, we can solve \wsp\ for $W$ by solving \wsp\ for $W_{\rm hard}$, which has no more than $k$ steps and each step has fewer than $k$ authorized users.
Hence, there can be no more than $k(k-1)$ authorized users in $W_{\rm hard}$.
\end{proof}


\begin{corollary}
$\wspi{1}(\ne)$ can be solved in time $\widetilde{O}(2^k)$.
\end{corollary}

\begin{proof}
The result follows immediately from Theorem~\ref{thm:1level}, the fact that there can be no more than $O(k^2)$ Type 1 constraints, and the proposition above.
\end{proof}

\begin{proposition}\label{pro:wsp1-kernel}
$\wspi{1}(\ne,=)$ has a kernel with at most $k(k-1)$ users.
\end{proposition}

\begin{proof}
The basic idea is to merge all steps that are related by constraints of the form $(=,s_1,s_2)$ for $s_1,s_2 \in S$.
More formally, consider an instance $\cal I$ of $\wspi{1}(=,\ne)$, given by a workflow $(S,U,\leqslant,A,C)$.
\begin{enumerate}[(1)]
  \item Construct a graph $H$ with vertices $S$, in which $s',s''\in S$ are adjacent if $C$ includes a constraint $(=,s',s'')$.
  \item If there is a connected component of $H$ that contains both $s'$ and $s''$ and $C$ contains a constraint $(\ne,s',s'')$ then $\cal I$ is unsatisfiable, so we may assume there is no such connected component. 
  \item For each connected component $T$ of $H$,
    \begin{enumerate}[(a)]
      \item replace all steps of $T$ in $S$ by a ``superstep'' $t$;
      \item for each such superstep $t$, authorize user $u$ for $t$ if and only if $u$ was authorized (by $A$) for all steps in $t$
      \item for each such superstep $t$, merge all constraints for steps in $t$.
    \end{enumerate}
\end{enumerate}
Clearly, we now have an instance of $\wspi{1}(\ne)$, perhaps with fewer steps and a modified authorization relation, that is satisfiable if and only if $\cal I$ is satisfiable.
The result now follows by Proposition~\ref{pro:wsp-ne-kernel}.

The reduction can be performed in time $O(kc + kn)$, where $c$ is the number of constraints: step (1) takes time $O(k+c)$; step (3) performs at most $k$ merges; each merge takes $O(k+c+n)$ time (since we need to merge vertices, and update constraints and the authorization relation for the new vertex set);\footnote{We can check step (2) when we merge constraints in step 3(c).} finally, if $k \leqslant c$ we have $O(k(k+c+n) = O(k(c+n))$, and if $c \leqslant k$ then we perform no more than $c$ merges in time $O(c(k+c+n)) = O(ck+cn) = O(ck+kn)$.
\end{proof}
\end{newstuff}

\begin{theorem}
$\wspi{1}(=,\ne)$ admits a kernel with at most $k$ users.
\end{theorem}

\begin{newstuff}
\begin{proof}
We first use the $\wspi{1}$ constraint reduction method from the proof of Proposition~\ref{pro:wsp1-kernel} to eliminate all constraints of the form $(=,s',s'')$, leaving an instance $\cal I$ of $\wspi{1}(\ne)$.
We now construct a bipartite graph $G=(U,S;A)$, where $A \subseteq S \times U$ is the authorization relation.
We may assume that $|U|\geqslant |S|=k$.

Let $V=U\cup S$. Using the well-known Hopcroft-Karp algorithm, we can find a maximum matching $M$ in $G$ in time $O(\sqrt{|V|}|A|)$.%
\footnote{A matching in a bipartite graph is a set of edges that are pairwise non-adjacent.
		  A maximum matching contains the largest possible number of edges.}
If $M$ covers every vertex of $S$, then $\cal I$ is satisfiable and our kernel is the subgraph of $G$ induced by all vertices covered by $M$.
(Since there is at most one edge in $M$ for each vertex in $S$ and at most one edge for each vertex in $U$, there are exactly $k$ users covered and we have a kernel containing $k$ users.)

If $M$ does not cover every vertex of $S$ then we define $R_{G,M}$ to be the set of vertices of $G$ which can be reached from some uncovered vertex in $S$ by an $M$-alternating path.%
\footnote{An $M$-alternating path has the property that for any pair of successive edges one belongs to $M$ and the other does not.}
Then a result of~\citeN[Lemma 3]{Sze2004} asserts that we can compute $R_{G,M}$ in time $O(|U|+|S|+|A|)$.
We write $R_{G,M}$ in the form $U' \cup S'$ for some $U' \subseteq U$ and $S' \subseteq S$.
The set $U' \cup S'$ has the following properties~\cite[Lemma 3]{Sze2004}:
\begin{compactenum}[P1.]
  \item All vertices of $S \setminus S'$ are covered by $M$;
  \item There is no edge in $G$ from $U \setminus U'$ to $S'$ and no edge of $M$ joins vertices in $U'$ with vertices in $S \setminus S'$;
  \item In the subgraph $G$ induced by $U' \cup S'$, vertices of a set $U'' \subseteq U'$ have at least $|U''|+1$ neighbors in $S'$.
\end{compactenum}
A bipartite graph $G$, a maximum matching $M$ in $G$ (indicated by the thicker lines), and the sets $U'$ and $S'$ are shown in Figure~\ref{fig:kernel-from-matching}; the figure is based on one used by~\citeN{Sze2004}.

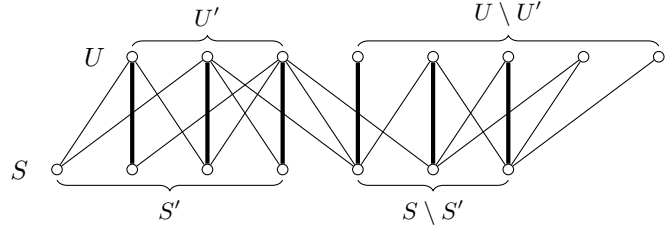
\begin{figure}[h]\centering
  \begin{tikzpicture}%
	[auto,circ1/.style={circle,draw,minimum size=4pt,fill=black!40,inner sep=0pt},%
	 circ2/.style={circle,draw,minimum size=4pt,inner sep=0pt}]
	\node[circ2] (s1) at (1,0) {};
	\node[circ2] (s2) at (2,0) {};
	\node[circ2] (s3) at (3,0) {};
	\node[circ2] (s4) at (4,0) {};
	\node[circ2] (s5) at (5,0) {};
	\node[circ2] (s6) at (6,0) {};
	\node[circ2] (s7) at (7,0) {};
	\node[circ2] (u1) at (2,1.5) {};
	\node[circ2] (u2) at (3,1.5) {};
	\node[circ2] (u3) at (4,1.5) {};
	\node[circ2] (u4) at (5,1.5) {};
	\node[circ2] (u5) at (6,1.5) {};
	\node[circ2] (u6) at (7,1.5) {};
	\node[circ2] (u7) at (8,1.5) {};
	\node[circ2] (u8) at (9,1.5) {};
	\node (U) at (1.5,1.5) {$U$};
	\node (S) at (0.5,0) {$S$};
	\draw[thin] (s1) -- (u1);
	\draw[thin] (s1) -- (u2);
	\draw[ultra thick] (s2) -- (u1);
	\draw[thin] (s2) -- (u3);
	\draw[thin] (s3) -- (u1);
	\draw[ultra thick] (s3) -- (u2);
	\draw[thin] (s3) -- (u3);
	\draw[ultra thick] (s4) -- (u3);
	\draw[thin] (s4) -- (u2);
	\draw[ultra thick] (s5) -- (u4);
	\draw[thin] (s5) -- (u5);
	\draw[thin] (s5) -- (u2);
	\draw[thin] (s5) -- (u3);
	\draw[thin] (s6) -- (u3);
	\draw[ultra thick] (s6) -- (u5);
	\draw[thin] (s6) -- (u6);
	\draw[thin] (s6) -- (u7);
	\draw[thin] (s7) -- (u5);
	\draw[ultra thick] (s7) -- (u6);
	\draw[thin] (s7) -- (u7);
	\draw[thin] (s7) -- (u8);
	\draw [decorate,decoration={brace,amplitude=4pt},yshift=4pt] (2,1.5) -- (4,1.5) node [black,midway,yshift=4pt] {\small $U'$};
	\draw [decorate,decoration={brace,amplitude=4pt},yshift=4pt] (5,1.5) -- (9,1.5) node [black,midway,yshift=4pt] {\small $U \setminus U'$};
	\draw [decorate,decoration={brace,amplitude=4pt},yshift=-4pt] (4,0) -- (1,0) node [black,midway,yshift=-4pt] {\small $S'$};
	\draw [decorate,decoration={brace,amplitude=4pt},yshift=-4pt] (7,0) -- (5,0) node [black,midway,yshift=-4pt] {\small $S \setminus S'$};
  \end{tikzpicture}
\caption{Constructing a kernel for $\wspi{}$ using a maximum matching}\label{fig:kernel-from-matching}
\end{figure}

Hence, we can assign users to all steps that are not in $S'$ (using $M$) and we will not violate any separation-of-duty constraints by doing so. Moreover, property (P2) means that allocating users in $U'$ to steps in $S'$ will not violate any separation-of-duty constraints.  In other words, we have reduced the problem instance to finding a solution to a smaller instance (the kernel) in which the set of users is $U'$, the set of steps is $S'$, and $|U'| < |S'| \leqslant k$.
\end{proof}

The authorization relation $A \subseteq S \times U$ defines the bipartite graph used to construct the matching.
The computation of a maximum matching in time $O(\card{A}\cdot\sqrt{n+k}) = O(nk\sqrt{n+k})$ enables us to compute a partial plan $\pi$, where an edge in the matching corresponds to a step $s$ and a user $u = \pi(s)$.
If the maximum matching has cardinality $k$, then we are done.
Otherwise, we solve \wsp\ for the kernel.

When the cardinality of $A$ is high (so the computation of the maximum matching is relatively slow), many users are authorized for many steps.
In this case, therefore, the observation that only those steps for which fewer than $k$ users are authorized need to be considered may mean that it is easy to decide whether the instance is satisfiable.
\end{newstuff}

\begin{newstuff}
We now state some negative results, negative in the sense that they assert that certain instances of \wsp\ do not have polynomial-size kernels.
The proofs of these results can be found in the appendix.
\end{newstuff}

\begin{theorem}\label{thm:wsp2equals-no-kernel}
$\wspi{2}(=)$ does not admit a kernel with a polynomial number of users unless $\conp \subseteq \np/\poly$.
\end{theorem}

\begin{newstuff}
\begin{theorem}\label{thm:wspcounting-no-kernel}
WSP with counting constraints of the type $(2,t,S')$ does not admit a kernel with a polynomial number of users unless $\conp \subseteq \np/\poly$.
\end{theorem}
\end{newstuff}

\begin{newstuff}
The above results tells us that there may be little to be gained from preprocessing an instance of $\wspi{2}(=)$ or an instance that contains arbitrary counting constraints, and we may simply apply the techniques described in Section~\ref{sec:fpt-entailment-constraints}.
Our final result in this section proves that the existence of a polynomial kernel is unlikely when we consider \wsp\ for canonical organizational hierarchies, even when we restrict attention to Type 1 constraints and hierarchies with only three levels.
\end{newstuff}

\begin{theorem} \label{no_pol_ker}
The problem $\wspi{1}(=,\ne,\sim,\nsim)$, where $\sim$ is an equivalence relation defined on $U$, does not have a polynomial kernel, unless $\np \subseteq \conp/\poly$.
\end{theorem}

\section{Concluding Remarks}

\begin{newstuff}
In general terms, the results reported in this paper provide a much improved understanding of the fixed parameter tractability of the workflow satisfiability problem.
In particular, we have developed a technique---the reduction of WSP to {\sc Max Weighted Partition}---that guarantees an instance of WSP is FPT, provided all constraints satisfy two simple criteria.
This enables the designer of workflow systems to determine whether the satisfiability of a workflow specification is FPT by examining the constraints defined in the specification.
Our results in this paper achieve several specific things.
\begin{itemize}
  \item First, the use of the {\sc Max Weighted Partition} problem to solve \wsp\ allows us to develop a fixed-parameter algorithm for which the worst-case run-time is significantly better than known algorithms.
  \item Second, this algorithm can be used to solve more general constraints---counting constraints, Type 3 entailment constraints and constraints based on equivalence relations---than was possible with existing methods.
      In short, we have extended the classes of workflow specifications for which the satisfiability problem is known to be FPT.
  \item Third, we have established the circumstances under which an instance of \wsp\ has a polynomial kernel.
        As well as providing the first results of this type for \wsp, kernelization is of enormous practical value.
        The computation of a maximum matching in time $O(nk \sqrt{n+k})$ is an extremely useful technique for deriving a (partial) plan for an instance of \wsp.
        Moreover, the reduction in the size of the problem instance when the maximum matching generates a partial plan will significantly reduce the complexity of solving instances of $\wspi{1}(=,\ne)$.
  \item Finally, we have significantly extended our understanding of those instances of \wsp\ that are FPT.
        Specifically, \wsp\ is FPT for any workflow specification that only includes constraints that are regular and for which (in)eligibility can be determined in time polynomial in the number of steps.
        In particular, we have established that \wsp\ problems which include constraints based on counting constraints and on user equivalence classes---enabling us to model organizational structures and business rules defined in terms of those structures---are still FPT.
\end{itemize}
In short, we believe our results represent a significant step forward in our understanding of the complexity of WSP and provide the blueprints for algorithms that can find efficient solutions for many practical instances of WSP.
\end{newstuff}

\subsection{Related Work}

Work on computing plans for workflows that must simultaneously satisfy authorization policies and constraints goes back to the seminal paper of \citeN{BeFeAt99}.
This work considered linear workflows and noted the existence of an exponential algorithm for computing valid plans.

Crampton extended the model for workflows to partially ordered sets (equivalently, directed acyclic graphs) and to directed acyclic graphs with loops~\cite{cram:sacmat05}.
Wang and Li further extended this model to include Type $2$ constraints and established the computational complexity and, significantly, the existence of fixed-parameter tractable algorithms for \mbox{$\wspi{2}(=,\ne)$}~\cite{WangLi10}.
Moreover, they established that $\wspi{2}$ is W[1]-hard, in general.

Recent work by~\citeN{BaBuKa10}  introduces the notion of \emph{release points} to model certain types of workflow patterns and defines the concept of \emph{obstruction}, which is related to the notion of unsatisfiability.
They prove that the \emph{enforcement process existence problem} (EPEP), which is analogous to \wsp\ for this extended notion of unsatisfiability, is NP-hard with complexity doubly-exponential in the number of users and constraints.

Independently of the work on authorization in workflows, there exists a vast literature on constraint satisfaction problems.
In this context, \cite{FeFrHeNaRo11} studied $\wspi{1}(\ne)$ and proved that this problem is fixed-parameter tractable.

Our work improves on that of Wang and Li and of Fellows {et al.} by establishing a tighter bound on the exponential factor of the fixed-parameter complexity for the relevant instances of \wsp\ (Theorem~\ref{thm:1level}).
Moreover, our work establishes that it is unlikely that our bound can be significantly improved (Theorem~\ref{thm:lb}).
We extend the type of constraints that can be defined by introducing counting constraints and Type $3$ entailment constraints, and we have shown that \wsp\ remains fixed-parameter tractable (Theorems~\ref{thm:counting-constraints-fpt} and~\ref{thm:type3}).

Most recently, we showed how WSP for entailment constraints could be reduced to {\sc Max Weighted Partition} for particular constraint relations.
In this paper, we have extended our approach to include any form of constraint that is regular and for which eligibility can be determined in time polynomial in the number of steps.
This represents a significant advance as it means we need only test whether a constraint is regular and devise an efficient eligibility test to deploy our techniques for solving WSP.

\subsection{Future Work}

There are many opportunities for further work in this area, both on the more theoretical complexity analysis and on extensions of \wsp\ to richer forms of workflows.
In particular, we hope to identify which security requirements can be encoded using constraints that satisfy the criteria identified in Theorem~\ref{thm:master}.
\begin{newstuff}
A very natural relationship between users is that of seniority: we would like to establish whether the inclusion of constraints based on this binary relation affects the fixed-parameter tractability of \wsp.
\end{newstuff}

There exists a sizeable body of work on \emph{workflow patterns}.
Many workflows in practice require the ability to iterate a subset of steps in a workflow, or to branch (so-called OR-forks and AND-forks) and to then return to a single flow of execution (OR-joins and AND-joins)~\cite{AaHof03}.
A variety of computational models and languages have been used to represent such workflows, including Petri nets and temporal logic.
To our knowledge, the only complexity results for richer workflow patterns are those of Basin {et al.} described above, which can handle iterated sub-workflows.
We will consider the fixed-parameter tractability of EPEP, and \wsp\ for richer workflow patterns, in our future work.

\begin{newstuff}
Wang and Li also introduced the notion of workflow \emph{resiliency}.
The \emph{static} $t$-resiliency checking problem (SRCP) asks whether a workflow specification remains satisfiable if some subset of $t$ users is absent.
Clearly SRCP is NP-hard as the case $t = 0$ corresponds to \wsp.
Evidently, SRCP can be resolved by considering the $\binom{n}{t}$ instances of \wsp\ that can arise when $t$ users are absent.
Hence, SRCP is in $\text{coNP}^{\text{NP}}$~\cite[Theorem 13]{WangLi10}.
The problems of deciding whether a workflow has \emph{dynamic} or \emph{decremental} $t$-resiliency are PSPACE-complete~\cite[Theorems 14--15]{WangLi10}.
\citeN{BaBuKa12} study a related problem called the \emph{optimal workflow-aware authorization administration problem}, which determines whether it is possible to modify the authorization relation, subject to some bound on the ``cost'' of the changes, when the workflow is unsatisfiable.
It will be interesting, therefore, to explore whether we can better understand the parameterized complexity of these kinds of problems.
\end{newstuff}

\appendix
\section{Proofs of Theorems}

In this appendix, we provide proofs of Theorems~\ref{thm:lb},~\ref{thm:wsp2equals-no-kernel},~\ref{no_pol_ker} and~\ref{thm:type3}.
Before proving Theorem~\ref{thm:lb}, we define two problems related to {\sc 3-Sat} and state two preparatory lemmas.
\begin{center}
\fbox{%
        \begin{tabulary}{.95\textwidth}{@{}r<{~}@{}L@{}}
          \multicolumn{2}{@{}l}{\sc $c$-Linear-3-Sat}\\
          \emph{Input:} & A $3$-CNF formula $\phi$ with $m$ clauses, and $n$ variables such that $m\leqslant cn$, where $c$ is a positive integer. \\
          \emph{Output:} & Decide whether there is a truth assignment satisfying $\phi$. \\
        \end{tabulary}%
      }
\end{center}

Let $\phi$ be a CNF formula. A truth assignment for $\phi$ is a {\em NAE-assignment} if, in each clause, it sets at least one literal true and at least one literal false. We say $\phi$ is {\em NAE-satisfiable} if there is a {\em NAE-assignment} for $\phi$.

\begin{center}
\fbox{%
        \begin{tabulary}{.95\textwidth}{@{}r<{~}@{}L@{}}
          \multicolumn{2}{@{}l}{\sc Not-All-Equal-3-Sat (NAE-3-Sat)}\\
          \emph{Input:} & A CNF formula $\phi$  in which every clause has exactly three literals.\\
          \emph{Output:} & Decide whether $\phi$ is NAE-satisfiable. \\
        \end{tabulary}%
      }
\end{center}

The first of our lemmas, which we state without proof, is due to Impagliazzo {\em et al.}~\cite{ImPaZa01} (see also~\cite{CoGuJoRaSa}).

\begin{lemma}\label{lem:lsat}
Assuming the Exponential Time Hypothesis, there exist a positive integer $L$ and a real number $\delta > 0$ such that {\sc $L$-Linear-3-SAT} cannot be solved in time $O(2^{\delta n})$.
\end{lemma}

\begin{lemma}\label{lem:NAE}
Assuming the Exponential Time Hypothesis, there exists a real number $\epsilon > 0$ such that {\sc NAE-3-SAT} with $n$ variables cannot be solved in time $O(2^{\epsilon n})$, where $n$ is the number of variables.
\end{lemma}

\begin{proof}
Let $L$ be an integer and $\delta$ be a positive real such that {\sc $L$-Linear-3-SAT} cannot be solved in time $O(2^{\delta n})$. Such constants $L$ and $\delta$ exist by Lemma \ref{lem:lsat}.
Suppose we have a polynomial time reduction from {\sc $L$-Linear-3-SAT} to {\sc NAE-3-SAT} and a positive integer $c'$ such that if a formula in {\sc $L$-Linear-3-SAT} has $n$ variables then the corresponding formula in {\sc NAE-3-SAT} has $n'$ variables and $n'\le c'n$.  Let $\epsilon = \delta/c'$ and suppose that {\sc NAE-3-SAT}  can be solved in time $O(2^{\epsilon n'})$, where $n'$ is the number of variables. Then {\sc $L$-Linear-3-SAT} can be solved in time $O(2^{\epsilon n'})=O(2^{\delta n})$, a contradiction to the definition of $\delta$.

It remains to describe the required polynomial time reduction from {\sc $L$-Linear-3-SAT} to {\sc NAE-3-SAT}. Recall that for every formula in {\sc $L$-Linear-3-SAT} we have $m\le Ln$, where $m$ and $n$ are the numbers of clauses and variables, respectively. We will show that our reduction gives $c'\le 2(1+L).$ Let $\phi$ be a formula of {\sc $L$-Linear-3-SAT}. Replace every clause $C=(u\vee v\vee w)$ in $\phi$ by
\begin{equation}\label{eq1} (u\vee v\vee x_C)\wedge (w\vee \overline{x_C}\vee y_C)\wedge (x_C\vee y_C\vee z)\end{equation}
to obtain a formula $\psi$ of {\sc NAE-3-SAT}.
Here variables $x_C$ and $y_C$ are new for every clause $C$ and $z$ is a new variable but it is common for all clauses of $\phi$.
We will show that $\phi$ is satisfiable if and only if $\psi$ is NAE-satisfiable.
This will give us $c'n\le n + 2m +1\le 2(1+L)n$ implying $c'\le 2(1+L).$

Let $V_{\phi}$ and $V_{\psi}$ be the sets of variables of $\phi$ and $\psi$, respectively. Hereafter 1 stands for {\sc true} and 0 for {\sc false}.

Assume that $\phi$ is satisfiable and consider a truth assignment $\tau:\ V_{\phi}\rightarrow \{0,1\}$ that satisfies $\phi$.  We will extend $\tau$ to $V_{\psi}$ such that the extended truth assignment is a NAE-assignment for $\psi$. We set $\tau(z)=1$. For each clause $C=(u\vee v\vee w)$ of $\phi$,
we set $\tau(y_C)=0$ and $\tau(x_C)=1-\max\{\tau(u),\tau(v)\}.$ Consider (\ref{eq1}).
 Since $\tau(y_C)=0$ and $\tau(z)=1$, $\tau$ is a NAE-assignment for the third clause in (\ref{eq1}). Since $\max\{\tau(u),\tau(v)\}\neq \tau(x_C)$, $\tau$ is a NAE-assignment for the first clause of (\ref{eq1}). Also, $\tau$ is a NAE-assignment for the second clause of (\ref{eq1}) because either $\tau(x_C)=\tau(y_C)=0$ or $\tau(u)=\tau(v)=0$ and, hence, $\tau(w)=1$.

Now assume that $\psi$ is NAE-satisfiable and consider a NAE-assignment $\tau:\ V_{\psi}\rightarrow \{0,1\}$ for $\psi$. Since
$\tau':\ V_{\psi}\rightarrow \{0,1\}$ is a NAE-assignment for $\psi$ if and only if so is $\tau''(t)=1-\tau'(t)$, $t\in  V_{\psi}$, we may assume that $\tau(z)=1.$ Since $\tau$ is a NAE-assignment for the third clause of (\ref{eq1}), we have $\min\{\tau(x_C),\tau(y_C)\}=0.$ If $\tau(x_C)=0$ then $\max\{\tau(u),\tau(v)\}=1$; otherwise $\tau(x_C)=1$ and $\tau(y_C)=0$ implying that $\tau(w)=1.$ Therefore, either $\max\{\tau(u),\tau(v)\}=1$ or $\tau(w)=1$ and, thus, $C$ is satisfied by $\tau$.
\end{proof}

\begin{proof}[of Theorem~\ref{thm:lb}]
Consider a CNF formula $\phi$, which is an instance of {\sc NAE-3-SAT}. Let $\{s_1,\ldots ,s_n\}$ be the variables of $\phi$ and let us denote the negation of $s_i$ by $s_{i+n}$ for each $i\in [n].$ For example, a clause $(s_1\vee \overline{s_2} \vee \overline{s_3})$ will be written as $(s_1\vee s_{n+2}\vee s_{n+3}).$
For $j\in [2n]$, we write $s_j=1$ if we assign {\sc true} to $s_j$ and $s_j=0$, otherwise.

Now we construct an instance of {\sc WSP}. The set of steps is $\{s_1,\ldots ,s_k\}$, where $k=2n$, and there are two users, $u_0$ and $u_1$.
We will assign user $u_i$ to a step $s_j$ if and only if $s_j$ is assigned $i$ in $\phi$.
For each $j\in [n]$ we set constraint $(\neq , s_j,s_{j+n})$. For every clause of $\phi$ with literals $s_{\ell},s_p,s_q$ we set constraint $(\neq , s_{\ell},  \{s_p,s_q\})$. We also assume that each user can perform every step subject to the above constraints.

Observe that the above instance of {\sc WSP} is satisfiable if and only if $\phi$ is NAE-satisfiable. Thus, we have obtained a polynomial time reduction of {\sc NAE-3-SAT} to {\sc WSP} with $\neq$ being the only binary relation used in the workflow and with just two users.
Now our theorem follows from Lemma~\ref{lem:NAE}.
\end{proof}

Before proving Theorem~\ref{thm:wsp2equals-no-kernel}, we introduce a definition and result due to \citeN{BTY09}.

\begin{definition}
Let $P$ and $Q$ be parameterized problems. We say a polynomial time computable function $f: \Sigma^* \times \mathbb{N} \rightarrow \Sigma^* \times \mathbb{N}$ is a \emph{polynomial parameter transformation} from $P$ to $Q$ if there exists a polynomial $p: \mathbb{N} \rightarrow \mathbb{N}$ such that for any $(x,k)\in \Sigma^* \times \mathbb{N}, (x,k) \in P$ if and only if $f(x,k)=(x',k') \in Q$, and $k' \le p(k)$.
\end{definition}

\begin{lemma}{\rm \cite[Theorem 3]{BTY09}}\label{ppt}
Let $P$ and $Q$ be parameterized problems, and suppose that $P^c$ and
$Q^c$ are the derived classical problems (where we disregard the parameter). Suppose that $P^c$ is NP-complete, and $Q^c \in \np$.
 Suppose that $f$ is a polynomial parameter transformation from $P$
to $Q$. Then, if $Q$ has a polynomial-size kernel, then $P$ has a polynomial-size kernel.
\end{lemma}

\begin{oldstuff}
\begin{proof}[of Theorem~\ref{thm:wsp2equals-no-kernel}]
In the {\sc Hitting Set} problem, we are given a set $Z$ and a family
${\cal F} = \{Z_1, \ldots , Z_m\}$ of subsets of $Z$ and are asked whether there exists a size-$k$
subset $H$ of $Z$ such that, for every $Z_i$ of $\cal F$, $H \cap Z_i \neq  \emptyset$. \citeN[Lemma 4]{WangLi10} gave a polynomial time reduction from {\sc Hitting Set} parameterized by $m+k$ to $\wspi{2}(=)$.
In fact, this is a polynomial parameter transformation.
\citeN{DLS09} proved that {\sc Hitting Set} parameterized by $m+k$ does not admit a polynomial-size kernel unless $\conp \subseteq \np/\poly$.
Now we are done by Lemma \ref{ppt}.
\end{proof}
\end{oldstuff}

\begin{newstuff}
\begin{proof}[of Theorem~\ref{thm:wsp2equals-no-kernel}]
We may formulate the {\sc Hitting Set} problem as a problem for bipartite graphs.
We are given a bipartite graph with with partite sets $U=\set{u_1,\dots ,u_n}$ and $V=\set{v_1,\dots ,v_m}$ and edge set $E$.
We are to decide whether there is a subset $H$ of $U$ with at most $k$ vertices such that each $v\in V$ has a neighbor in $H$.

We say that two problems are {\em equivalent} if every yes instance of one corresponds to a yes instance of the other.
\citeN[Lemma 4]{WangLi10} proved that  {\sc Hitting Set} is equivalent to the following subproblem $\Pi$ of $\wspi{2}(=)$.
We have $U$ as the set of users, $V\cup S$ as the set of $k'=m+k$ steps, every user $u_j$ is authorized to perform any step from $S$ and every step $v_i$ such that $u_jv_i\in E$, and $(=,v_i,S)$, $i\in [m]$, is the set of constraints of Type 2.

Observe that the above construction gives a polynomial parameter transformation from {\sc Hitting Set} parameterized by $m+k$ to $\wspi{2}(=)$.
\citeN{DLS09} proved that {\sc Hitting Set} parameterized by $m+k$ does not admit a polynomial-size kernel unless $\conp \subseteq \np/\poly$.
Now we are done by Lemma~\ref{ppt}.
\end{proof}

\begin{proof}[of Theorem~\ref{thm:wspcounting-no-kernel}]
We will use the polynomial parameter transformation from {\sc Hitting Set} parameterized by $m+k$ to a subproblem $\Pi$ of WSP described in the proof of Theorem~\ref{thm:wsp2equals-no-kernel}.
We obtain a subproblem $\Pi^*$ of WSP with counting constraints of the type $(2,t,S')$ from $\Pi$ by keeping the same set $U$ of users and the same set $V\cup S$ of steps, but by replacing the constraints of $\Pi$ with  $(2,k+1,S\cup \set{v_i})$, $i\in [m]$.

We now prove that $\Pi$ and $\Pi^*$ are equivalent, from which the result follows by Theorem~\ref{thm:wsp2equals-no-kernel}.

Let $\pi^*$ be a valid plan for $\Pi^*$ and let $\pi$ be obtained from $\pi^*$ by restricting it to $V\cup S$.
Observe that if a constraint $(2,k+1,S\cup \set{v_i})$ is satisfied by $\pi^*$, then $(=,v_i,S)$ is satisfied by $\pi$.
Thus, $\pi$ is a valid plan for $\Pi$.

Let $\pi$ be a valid plan for $\Pi$ and let $\pi^*$ be obtained from $\pi$ by  reassigning to $\pi(v_1)$ every step $s$ in $S$ such that the user $\pi(s)$ is assigned to perform just one step in $V\cup S$.
Observe that if $(=,v_i,S)$ is satisfied by $\pi$, then $(2,k+1,S \cup \set{v_i})$ is satisfied by $\pi^*$.
Thus, $\pi^*$ is a valid plan for $\Pi^*$.
\end{proof}
\end{newstuff}

The following two definitions and Theorem~\ref{res} are due to \citeN{BodJanKra}.

\begin{definition}[Polynomial equivalence relation]\label{polequ}
An equivalence relation $\mathcal{R}$ on $\Sigma^\ast$ is called a {\em polynomial equivalence relation} if the following two conditions hold:
\begin{compactitem}
   \item{There is an algorithm that given two strings $x,y\in\Sigma^\ast$ decides whether $x$ and $y$ belong to the same equivalence class
in $(|x|+|y|)^{O(1)}$ time.}
    \item{For any finite set $S\subseteq\Sigma^\ast$ the equivalence relation $\mathcal{R}$ partitions the elements of $S$ into at most
$(\max_{x\in S}|x|)^{O(1)}$ equivalence classes.}
\end{compactitem}
\end{definition}

\begin{definition}[Cross-composition]\label{crocom}
Let $L\subseteq\Sigma^\ast$ be a problem and let $Q\subseteq\Sigma^\ast\times\mathbb{N}$ be a parameterized problem. We say that $L$
{\em cross-composes} into $Q$ if there is a polynomial equivalence relation $\mathcal{R}$ and an algorithm which, given $t$ strings $x_1,\dots,x_t$
belonging to the same equivalence class of $\mathcal{R}$, computes an instance $(x^\ast,k^\ast)\in\Sigma^\ast\times\mathbb{N}$ in time
polynomial in $\sum_{i=1}^t|x_i|$ such that:
\begin{compactitem}
   \item {$(x^\ast,k^\ast)\in Q$ if and only if $x_i\in L$ for some $1\leqslant i\leqslant t$.}
  \item{$k^\ast$ is bounded by a polynomial in $\max_{i=1}^t|x_i|+\log t$.}
\end{compactitem}

\end{definition}

\begin{theorem}\label{res}
If some problem $L$ is $\np$-hard under Karp reductions and $L$ cross-composes into the
parameterized problem $Q$ then there is no polynomial kernel for $Q$ unless $\np \subseteq \conp/\poly$.
\end{theorem}

\begin{proof}[of Theorem~\ref{no_pol_ker}]
We may treat $\wspi{1}(=,\ne,\sim,\nsim)$ as an instance of $\wspi{1}(\sim_1,\nsim_1,\sim_2,\nsim_2,\sim_3,\nsim_3)$ for a canonical hierarchy with three levels, where $\sim_1$ and $\nsim_1$ correspond to $=$ and $\ne$ respectively.

We will use Theorem \ref{res} to show the result, hence we need an $\np$-hard problem $L$ which cross-composes into $\wspi{1}(\sim_1,\not\sim_1,\sim_2,\not\sim_2,\sim_3,\not\sim_3)$.
For this purpose we will use the problem $3$-{\sc Coloring}.
An instance of $3$-{\sc Coloring} is a graph $G$ in which we want to decide if it can be $3$-colored.
We say that two graphs $G_1$ and $G_2$ are equivalent if $|V(G_1)|=|V(G_2)|$.
It is not difficult to see that this defines a polynomial
equivalence relation on $3$-{\sc Coloring} (see Definition \ref{polequ}).

Consider now $t$ instances of $3$-{\sc Coloring}, $G_1,G_2,\ldots,G_t$. Let
  \[
    k=|V(G_1)|=|V(G_2)|= \cdots = |V(C_t)|\quad\text{and}\quad V(G_i)=\{x_1^i,x_2^i,\ldots,x_k^i\}.
  \]
We now construct an instance of \mbox{$\wspi{1}(\sim_1,\not\sim_1,\sim_2,\not\sim_2,\sim_3,\not\sim_3)$} with steps $S$ and users $U$ defined as follows.
\[
  \begin{array}{rclcl}
  S & = & (\cup_{i=1}^k V_i) \cup_{1 \leqslant i < j \leqslant k} \{e_{i,j},e_{i,j}'\} & \hspace{1cm} & \mbox{where $V_i = \{v_1^i,v_2^i\}$}; \vspace{0.2cm} \\
  U & = & \cup_{i=1}^t U_i & &  \mbox{where $U_i = \{c_1^i,c_2^i,c_3^i,\alpha^i\}$}. \\
  \end{array}
\]
Observe that $|S| = k+k^2$ is bounded by a polynomial in the
maximum size of $G_i,\ i\in [t]$.

We now define the hierarchy ${\cal H} = U^{(1)}, U^{(2)}, U^{(3)}$, where $U^{(1)}$ is the partition of $S$ containing all singletons, $U^{(2)}$ is the partition $U_1,U_2,\ldots,U_t$ and
$U^{(3)}$ is the partition containing just one set $S$. We now define the constraints $C$ as follows.
\[
  \begin{array}{rclcl}
  C & = & \{ (\not\sim_1,v_1^i,v_2^i) \; | \; i\in [k] \} \cup  \\
  & & \{ (\not\sim_1,v_1^i,e_{i,j}), (\not\sim_1,v_2^i,e_{i,j}), (\not\sim_1,e_{i,j},e_{i,j}'), (\not\sim_1,v_1^j,e_{i,j}'), (\not\sim_1,v_2^j,e_{i,j}') \; | \; 1 \leqslant i < j \leqslant k \} \\
  & & \cup \{ (\sim_2,s_1,s_2) \; | \; s_1,s_2 \in S \} \\
\end{array}
\]

We now let all users, except $\alpha^1,\alpha^2,\ldots,\alpha^t$, be authorized for all steps. Furthermore if $x_i^a x_j^a \not\in E(G_a)$, where
$1 \leqslant i < j \leqslant k$ then authorize $\alpha^a$ for $e_{i,j}$.

\paragraph{Claim A} {\em The created instance has a valid plan if and only if one of the graphs, $G_1,G_2,\ldots,G_t$ are $3$-colorable.}  The result now follows by Theorem~\ref{res}.

\paragraph{Proof of Claim A} Assume that the created instance has a valid plan. The constraints $\{ (\sim_2,s_1,s_2) \; | \; s_1,s_2 \in S \}$ imply that
all users used in the plan belong to exactly one block in $U^{(2)}$, say $U_r$. Let $\gamma_j \in \{1,2,3\}$ be defined such that the users assigned
to $v_1^j$ and $v_2^j$ are $\{c_1^r,c_2^r,c_3^r\} \setminus \{c_{\gamma_j}^r\}$, which is possible as $(\not\sim_1,v_1^j,v_2^j) \in C$ and $\{c_1^r,c_2^r,c_3^r\}$
are the only users from $U_r$ authorized for $\{v_1^j,v_2^j\}$. If $x_i^rx_j^r \in E(G_r)$ then $e_{i,j}$ must be assigned user $c_{\gamma_i}^r$ and $e_{i,j}'$
must be assigned user $c_{\gamma_j}^r$, which implies that $\gamma_i \not= \gamma_j$, by the given constraints. Therefore $\gamma_1,\gamma_2,\ldots,\gamma_k$
is a $3$-coloring of $G_r$. This shows one direction of Claim A.

Now assume we have a $3$-coloring $\gamma_1,\gamma_2,\ldots,\gamma_k$ of $G_r$, for some $r \in [t]$. Assign users  $\{c_1^r,c_2^r,c_3^r\} \setminus \{c_{\gamma_j}\}$
to the steps $v_1^j$ and $v_2^j$ and for all $e_{i,j}$ assign user $c_{\gamma_i}^r$ if $x_i^r x_j^r \in E(G_r)$ or user $\alpha_r$ if $x_i^r x_j^r \not\in E(G_r)$.
Finally assign user $c_{\gamma_j}^r$ to all steps $e_{i,j}'$. Note that the given assignment of users satisfies all constraints, which completes the proof of the claim.
\end{proof}

\begin{proof}[of Theorem~\ref{thm:type3}]
The result follows from a very similar argument to that used in the proof of Theorem~\ref{thm:llevel}.
Notice that our method for identifying ineligible sets for Type $2$ constraints of the form $(\not\sim_i,s,S')$ works equally well for  Type $3$ constraints of the form $(\not\sim_i,S_1,S_2)$ (since a set $F$ is ineligible if $S_1 \cup S_2 \subseteq F$).

However, we cannot use our method for constraints in $C_{\sim}$.
Nevertheless, we can rewrite the set of constraints in $C_{\sim}$ as Type $2$ constraints, at the cost of introducing additional workflow steps (as we did in the proof of Theorem~\ref{thm:1level-type3}).
This requires the replacement of $c'$ Type $3$ constraints by $2c'$ Type $2$ constraints and the creation of $c'$ new steps.
Finally, we solve the resulting instance of $\wspi{2}$ for a workflow with $n$ users, $k+c'$ steps and $c + 2c'$ constraints, which has complexity $\widetilde{O}((c+2c')n2^{k+c'} + n^2 3^{k+c'})$, by Theorem~\ref{thm:llevel} and Remark~\ref{rem:optimized-proof}.

We may assume without loss of generality that for all constraints of the form \mbox{$(\sim_i,S_1,S_2)$} in $C_{\sim}$, $S_1 \cap S_2 = \emptyset$.
(The constraint is trivially satisfied if there exists $s \in S_1 \cap S_2$, since we assume there exists at least one authorized user for every step.)
Hence the number of constraints having this form is no greater than $\sum_{j=1}^k \binom{k}{j}2^{k-j} = 3^k$.
\end{proof}

\bibliography{refs}
\bibliographystyle{acmsmall}

\end{document}